\documentclass[12pt]{article}
\usepackage{geometry}
\usepackage{bbold}
\usepackage[utf8]{inputenc}
\usepackage[T1]{fontenc}
\usepackage{graphicx}
\usepackage{multirow}
\usepackage{hyperref}
\usepackage{lmodern}
\usepackage{siunitx}
\usepackage{longtable}
\usepackage{amsmath,amssymb,amsthm}
\usepackage{xspace}
\usepackage{placeins}

\newtheorem{proposition}{Proposition} 
\usepackage{algorithm}
\usepackage{algpseudocode}

\newcommand{\set}[1]{ [#1]}
\newcommand{\ZIPLN}{ZI-PLN\xspace}
\newcommand{\ZIPLNPCA}{ZI-PLN-PCA\xspace}
\newcommand{\PLNPCA}{PLN-PCA\xspace}

\usepackage{authblk}

\usepackage{listings}
\usepackage[dvipsnames]{xcolor}
\usepackage{colortbl}
\usepackage{mathtools}
\usepackage{mathrsfs}
\usepackage{url}
\usepackage[caption=false]{subfig}
\usepackage{caption}
\usepackage{titlesec}
\usepackage[bottom]{footmisc}
\usepackage{booktabs}
\usepackage{enumerate}
\usepackage{enumitem}
\usepackage{pdfpages}
\usepackage{comment}
\usepackage{adjustbox}
\usepackage{natbib} 
\usepackage{float}
\usepackage{tikz}
\usetikzlibrary{decorations.pathreplacing}

\newcommand{\argmax}{\mbox{argmax}} 
\newcommand{\Bcal}{\mathcal{B}} 
 
\newcommand{\Cov}{\mathbb{C}\text{ov}} 
\newcommand{\Esp}{\mathbb{E}} 
\newcommand{\Espt}{\widetilde{\Esp}}
\newcommand{\Hcal}{\mathcal{H}} 
\newcommand{\Jcal}{\mathcal{J}} 
\newcommand{\logit}{\operatorname{logit}} 
\newcommand{\Qcal}{\mathcal{Q}} 

\newcommand{\Mcal}{\mathcal{M}}

\newcommand{\Ncal}{\mathcal{N}} 
\newcommand{\Ocal}{\mathcal{O}} 
\newcommand{\One}{\mathbf{1}} 
\newcommand{\Pcal}{\mathcal{P}} 
\newcommand{\pt}{q} 
 
\newcommand{\Var}{\mathbb{V}} 

\renewcommand{\d}{\text{d}}

\newcommand{\Oi}{\Ocal_i} 
\newcommand{\Mi}{\Mcal_i} 
\newcommand{\probU}{\pi} 
\newcommand{\diag}{\text{diag}} 

\newtheorem{model}{Model}
\makeatletter
\@ifundefined{remark}{%
}{}
\makeatother

\title{\textbf{Inferring the presence and abundance of rare waterbirds species from scarce data}}

\author[1,2]{Barbara Bricout}
\author[2]{Laura Dami}
\author[3]{Pierre Defos du Rau}
\author[4]{Sophie Donnet}
\author[2]{Thomas Galewski}
\author[1]{Stéphane Robin}

\affil[1]{Sorbonne Université, Université Paris Cité, CNRS, LPSM, 75005 Paris, France}
\affil[2]{Tour du Valat, Institut de recherche pour la conservation des zones humides méditerranéennes, 13200 Arles, France}
\affil[3]{Office Français de la Biodiversité, DRAS, Service Conservation et Gestion Durable des Espèces Exploitées, 13200 Arles, France}
\affil[4]{Université Paris-Saclay, AgroParisTech, INRAE, UMR MIA Paris-Saclay, 91120 Palaiseau, France}

\date{}

\begin{document}

\maketitle

\begin{abstract}

Abundance data are used in ecology for species monitoring and conservation. These count data often display several specific characteristics like numerous missing data, high variance, and a high proportion of zeros, particularly when monitoring rare species.
We present a model that aims to impute missing data and estimate the effect of covariates on species presence and abundance. It is based on the log-normal Poisson model, which offers more flexibility in the variance of counts than a Poisson model. A latent variable is added for the overrepresentation of zeros in the data. The imputation of missing data is made possible by assuming that the latent variance matrix has low rank and the inclusion of covariates. \\
We demonstrate the identifiability in the presence of missing data. Since maximum likelihood inference is intractable, we use a variational expectation-maximization algorithm to infer the parameters. We provide an estimate of the asymptotic variance of the estimators and derive prediction intervals for the imputations, an estimate of the temporal trend, and a procedure for detecting a potential change in this trend. \\
We evaluate our imputations and associated prediction intervals using artificially degraded monitoring data set. We conclude with an illustration on a monitoring waterbirds data set.
    
\end{abstract}




\section{Introduction}


\paragraph{Context.}
Monitoring biodiversity is essential for assessing species conservation status and informing public policies aimed at protecting ecosystems. However, biodiversity monitoring programs often face limited material and human resources. As a result, abundance data collected across space and time are frequently incomplete—this is notably the case for winter waterbird surveys in North African countries.

Bird population monitoring data generally consists of recording the number of individuals of the species studied at a series of sites on different years, as well as information (covariates) describing the spatial and temporal environment.
Missing data can significantly hinder both the estimation of overall population sizes and the statistical analyses required to detect long-term trends. Consequently, conservation assessments may be biased or incomplete. This issue is particularly critical for rare species, which tend to exhibit a high proportion of null counts and substantial variance in their counts. In the case of waterbirds, a high variability is often driven by gregarious behavior and fluctuating local aggregation patterns.
In such situations, improving our ability to detect and interpret population trends is crucial, as it provides the foundation for designing and implementing appropriate conservation actions. These challenges underline the need for robust statistical approaches capable of addressing data incompleteness, facilitating trend analysis, and characterizing the ecological attributes of sites that support rare or irregularly observed species.

\paragraph{Related works.}
To impute missing values in large-scale ecological monitoring data sets  multivariate count data, \citet{robin2019low} introduced LORI, a new imputation method refered.  This method has shown improved accuracy over more usual approaches such as TRIM (\cite{dakki2021imputation}), particularly in contexts involving many sites and time periods. However, LORI does not model zeros and may not be adapted to species with a high proportion of zeros, a common feature in data sets involving rare species.

One potential solution to this issue is to model species presence and individual abundance separately. This is the idea behind zero-inflated models, which separates structural zeros—arising from true absence—and sampling zeros due to nondetection or low abundance. Zero-inflated Poisson (ZIP) and Zero-inflated Negative Binomial (ZINB: \cite{lambert1992zero}) models have been widely used in ecology to account for excess zeros \cite{martin2005zero,wenger2008estimating,agarwal2002zero}. However, most of these models remain univariate or assume independence across sampling units, limiting their capacity to capture spatial or temporal correlations. This is why several multivariate extensions of zero-inflated models have been proposed in recent years \cite{li1999multivariate, dong2014multivariate, lee2020bayesian}. In particular, \citet{BCG24} introduce a multivariate model built upon the Poisson Log-Normal (PLN) model (first introduced by \citet{AiH89}), which incorporates a latent Gaussian variable to jointly model overdispersion and correlations between variables, while also explicitly accounting for zero inflation. The resulting model is called \ZIPLN. Due to the presence of the Gaussian latent layer, the inference of the parameters of the variant of the PLN model we cannot rely on a standard Expectation-Maximization (EM) algorithms \cite{dempster1977maximum}, and variational inference approaches \cite{blei2017variational, chiquet2019variational, chaussard2025tree} are often adopted.

\paragraph{Our contribution.}
Our aim is to explicitly address the problem of missing data, which means developing a model that is suitable for imputation. We limit ourselves to the abundance of one species, observed in different sites at different times.
We consider a model specifically designed for overdispersed count data. Our approach builds on the Poisson Log-Normal (PLN) framework, but introduces temporal dependency across years by enforcing a low-rank structure on the latent Gaussian layer likewise \citet{chiquet2018variational}. In addition, we incorporate a latent Bernoulli variable to explicitly model species presence or absence, thus accounting for zero inflation. Because the observations are collected both in time and space our framework incorporates three types of covariates associated with sites, years, and site-and-year, respectively.


We determine the condition on the missing data pattern that ensures the identifiability of the model, in the case of both a full-rank and a low-rank covariance matrix.

With regard to inference, we also rely on a variational approximation of the EM algorithm and provide an approximation of the asymptotic variance of the estimators based on \citet{WM19} (and used by \cite{BCM24} in the context of PLN models). Another important contribution is the imputation of missing data, as well as a Monte Carlo procedure for calculating prediction intervals. We also provide an empirical study of the respective effect of different pattern of missing data on the inference. Finally, we propose a strategy to infer temporal tendencies together with change point in the temporal tendencies. 

\paragraph{Paper outline.} Section \ref{sec:model} presents the data and the model, which we prove to be identifiable in presence of missing data.  
In  Section \ref{ref:inference}, we provide the expressions for the implementation of the VEM algorithm together with a method to compute confidence intervals and imputation strategies.  
In Section \ref{sec:data}, we conduct a simulation study on waterbird counts collected in two European countries during 20 years, in which missing values are artificially introduced.
Finally, in the same section, we apply the model to an incomplete dataset from the International Waterbird Census (IWC) and show how the results can be used to estimate temporal trends.
The data we use is not public at this time, which is why species names are anonymized throughout the article.


\section{\ZIPLN model and partial observations}\label{sec:model}

\subsection{Data and notations} \label{sec:notations}

\paragraph{Partial observations.}
We consider a data set collected on $n$ sites and along $p$ years. Let  $Y_{ij}$ be the number of individuals of the species of interest counted in site $i$ ($1 \leq i \leq n$) on year $j$ ($1 \leq j \leq p$). $\Omega_{ij}$  is the  indicator variable being $1$ when site $i$ has been actually visited on year $j$ and $0$ otherwise (note that, when the site is visited: $\Omega_{ij} = 1$, the observed count $Y_{ij}$ may still be zero). 
Then, we denote by $Y_i$ the $p$-dimensional vector of counts in site $i$ across the $p$ years: $Y_i = [Y_{i1} \dots Y_{ip}]^\top$;  $Y = [Y_{ij}]_{1 \leq i \leq n, 1 \leq j \leq p}$ is the complete data matrix and $\Omega = [\Omega_{ij}]_{1 \leq i \leq n, 1 \leq j \leq p}$ the matrix of indicators, from which we define 
$\Ocal =  \{(i,j) : \Omega_{ij}  = 1\}$ the set of actually observed data and 
$\Mcal =  \{(i,j) : 1 \leq i \leq n, 1 \leq j \leq p\} \setminus \Ocal$ the set of missing data. 

The complete data set $Y $ can therefore be split into the set of observed counts $Y^O$ and the set of missing counts $Y^M$:
$$
Y^O = \{Y_{ij}: (i, j) \in \Ocal\}, \qquad
Y^M = \{Y_{ij}: (i, j) \in \Mcal\}.
$$
We denote by 
$\Oi$ the set of observed years in site $i$: $\Oi = \{1 \leq j \leq p: \Omega_{ij} = 1\}$, 
$\Mi$ the set of missing years in site $i$: $\Mi = \{1, \dots p\} \setminus \Oi$
and by $Y_i^O$ and $Y_i^M$ the vectors of observed and missing counts in site $i$, respectively: $Y_i^O = [Y_{ij}]_{j \in \Oi}$, $Y_i^M = [Y_{ij}]_{j \in \Mi}$.
We assume that each $\Oi$ is never empty (each site $i$ is observed at least one year $j$).  
\paragraph{Covariates.} 
The counts are enriched by covariates that can be decomposed as:  covariates on sites (encoded in a matrix $X_R$ with dimensions  $n \times d_1$ where the index $R$ stands for "rows"), covariates on years (encoded in a matrix $X_C$ with dimensions $p \times d_2$ where the index $C$ stands for "columns") and covariates specific to each couple sites/year (encoded in a matrix $X_E$ with dimensions $np \times d_3$). The three sets of covariates are put together in the covariate matrix $X = [ X_R \otimes \One_p  \quad \One_n \otimes X_C \quad X_E]$ with dimension $np \times d$ where $ d = d_1 + d_2 + d_3$. We denote by $x_{ij}$ the $d$-dimensional vector of covariates  describing site $i$ on year $j$ and we assume that all the covariates have been observed in each site on each year.


\paragraph{Ignorability.} 
Throughout the paper, we assume that the process generating the missing data pattern $\Omega$ is independent from the counts $Y$,  conditionally on the covariates $X$. More specifically, we assume that the joint conditional distribution $p(Y, \Omega \mid X)$ factorizes as $p(Y, \Omega \mid X) = p(\Omega \mid X) p(Y \mid X)$. Furthermore, we assume ignorability (as defined by \citet{little2019statistical}), that is to say, that these two conditional distribution are respectively controlled by independent sets of parameters, that is $p_{\eta, \theta}(Y, \Omega \mid X) = p_\eta(\Omega \mid X) p_\theta(Y \mid X)$ where $\theta$ is the parameter of interest. As a consequence, marginalizing over the non-observed counts $Y^M$, we have that
\begin{equation} \label{eq:ignor}
    p_{\eta, \theta}(Y^O, \Omega \mid X) = p_\eta(\Omega \mid X) p_\theta(Y^O \mid X),
\end{equation}
so we will only have to deal with $p_\theta(Y^O \mid X)$ for the inference of $\theta$. As a counterpart, the inference we are about to describe will not inform us about the conditional distribution of the missing data pattern $p_\eta(\Omega \mid X)$ in any way. \\
In the sequel, for the sake of clarity, we will drop the conditioning with respect to the covariates $X$ in all formulas and denote $p_\theta(Y)$ (resp. $p_\theta(Y^O)$) in place of $p_\theta(Y \mid X)$ (resp. $p_\theta(Y^O \mid X)$).


\paragraph{Imputation.} 
One of our main goal is to provide predictions for the missing counts $Y^M$ based on the covariates $X$ and the observed counts $Y^O$. The former type of dependence can typically be addressed through regression terms, whereas the latter necessarily relies on some relationship between the observed and unobserved responses.  To fit these requirements and to account some data specificities, which we will describe hereafter, we consider a generalization of the Poisson log-normal model \cite{AiH89}, which has become a generic framework for community ecology \cite{CMR21}. We first describe this extension in absence of missing data.

\subsection{The Poisson-log normal model with inflation of zeros} 


Analyzing dependant count data is a significant challenge that can be addressed by statistical model-based approaches. One such approach is the multivariate Poisson-Log-Normal (PLN) model, which assumes that counts are driven by an underlying latent Gaussian variable, with dependencies among counts arising solely from latent correlations. Furthermore, the latent covariance matrix can be assumed to have a low rank, which gives rise to the PLN-PCA model \citep{chiquet2018variational}.
However, the PLN model does not accommodate zero-inflation, as can be expected for the counts of a rare bird species.
To address this limitation, \citet{GQTMC2023} or \citet{BCG24} proposed the Zero-Inflated PLN (\ZIPLN) model, which augments PLN with a multivariate zero-inflated component via an additional Bernoulli latent variable. 

\ZIPLNPCA  relates species abundances $Y$ to covariates $X$, explicitly addressing three salient features of the data: (i) abundances are non-negative integer counts; (ii) a substantial proportion of the observations are zeros (Figure~\ref{fig:hist}) and (iii)  
dependence exists among observations collected at the same site across years as follows.

\begin{model}[\ZIPLNPCA model] \label{mod:ZIPLNPCA}
  Consider $n$ sites and $p$ years, the abundances $\{Y_{ij}\}_{i \in \set{n}, j\in \set{p}}$ (where $\set{n} = \{1, \dots,n\})$ are drawn as follows:
  \begin{itemize}
      \item the presence of the species in each site and year $\{U_{ij}\}_{i \in \set{n}, j\in \set{p}}$ are all independent with respective distribution
      \begin{equation}\label{eq:PLN1}
      U_{ij} \sim \Bcal\left(\pi_{ij}\right)
      \qquad \text{with} \quad \logit(\pi_{ij}) = x^\top_{ij}\gamma ,
       \end{equation}
      where $\logit(u) = \log[u/(1-u)]$ for $0 < u < 1$;
      \item the site-specific $q$-dimensional latent vectors $\{W_i\}_{i \in \set{n}}$ are all independent with standard Gaussian distribution $\Ncal(0, I_q)$, from which the $p$-dimensional vectors $\{Z_i\}_{i \in \set{n}}$ are respectively formed as  
      \begin{equation}\label{eq:PLN3} 
      Z_i = C W_i ;
      \end{equation} 
      \item the number of individuals in each site and year are independent conditional on the   presence indicators $\{U_{ij}\}_{i \in \set{n}, j \in \set{p}}$ and the site-specific vectors $\{W_i\}_{1 \leq i \leq n}$ with respective conditional distributions
      \begin{equation}\label{eq:PLN2}
      Y_{ij} \sim \delta_{\{0\}} \quad \text{if} \quad U_{ij} = 0 
      \qquad \text{and} \qquad 
      Y_{ij} \sim \Pcal(\exp(x_{ij}^\top \beta + Z_{ij})) \quad \text{if} \quad U_{ij} = 1,
      \end{equation}
      $\delta_{\{0\}}$ being the Dirac mass at zero.
  \end{itemize}
  The set of parameters of this model is $\theta = (\beta, \gamma, C)$. 
\end{model}
From Equation \eqref{eq:PLN1}, species presence at site $i$ in year $j$ is modeled through a Bernoulli process with probability $\pi_{ij}$ parameterized by the covariates $x_{ij}$; absence yields zero abundance in a deterministic manner. 
Second (Equation \eqref{eq:PLN2}), conditional on presence, the abundance at site $i$ and year $j$ $Y_{ij}$ follows a Poisson distribution whose mean depends on both the covariates $x_{ij}$ and a latent Gaussian term capturing within-site dependency structure $Z_{ij}$. Specifically, each site $i$ is associated with a $p$-dimensional Gaussian random vector $Z_i$ with non-diagonal covariance $\Sigma$, thereby inducing correlation across years. Third, to parsimoniously represent this dependence, the covariance matrix $\Sigma = [\sigma_{jk}]_{(j,k) \in \set{p}^2}$ is assumed to have rank $q \leq p$, such that $\Sigma = CC^\top$, with $C$ a $p \times q$ real matrix and $Z_i = C W_i$, where $W_i$ is a standard $q$-dimensional Gaussian vector (see Equation \eqref{eq:PLN3}). This low-rank formulation parallels the structure of probabilistic principal component analysis (PCA) \cite{TiB99}.


Note that, in Model \ref{mod:ZIPLNPCA}, neither the indicators $U_{ij}$ nor the vectors $W_i$ are observed: they constitute the latent variables of the model. 
As for its parameters, $\gamma$ encodes the respective effects of the covariates on the presence of the species, whereas $\beta$ encodes their effects on its abundance, and $C$ encodes the dependency structure between observations from the same site.
Observe that this dependency only acts on the abundances (conditional on the presence), and not on the presence itself.



Model~\ref{mod:ZIPLNPCA} is closely related to the framework introduced by \citet{BCG24}, but differs in several key aspects. Conceptually, while \cite{BCG24} consider  multiple (dependent) species observed across several sites, our model focuses on a single species monitored at multiple sites over several years. 
To this respect, in Model~\ref{mod:ZIPLNPCA} the different years play the same role as the different species in \cite{BCG24}. A consequence is that our  model must incorporate three types of covariates specific to sites, to years and to site--year (interaction), respectively, whereas only the first type of covariate is included in \cite{BCG24}. Obviously, this modifies the matrix expression of the (complete) likelihood.
However, the most substantial distinction  lies in the inference procedure: our observation setting requires estimation in the presence of missing data (from $Y^O$), a frequent occurrence in ecological studies, and we aim to predict unobserved abundances $Y^M$. 
The assumption of a low-rank covariance structure $\Sigma$ ensures robustness of its estimation, which is desirable in presence of numerous missing data.



The moments of the distribution defined in Model \ref{mod:ZIPLNPCA}
are easily computed and can be interestingly interpreted. 

\begin{proposition}\label{prop:moments}
    Under Model \ref{mod:ZIPLNPCA}, we have
\begin{align} \label{eq:esp_moments}
   \Esp(Y_{ij}) & = \pi_{ij} \exp(x_{ij}^\top \beta + \sigma_j^2/2) =: E_{ij} ,  \\
    \Var(Y_{ij}) & = E_{ij} +  E^2_{ij} \; (e^{\sigma_j^2}- \pi_{ij}) / \pi_{ij} , \\
    \Cov(Y_{ij}, Y_{ik}) &= E_{ij}  E_{ik} (e^{\sigma_{jk}} - 1) &  \text{for } j \neq k , 
\end{align}
where $\Sigma = (\sigma_{jk})_{(j,k)\in \set{p}^2}$ and $\sigma_j^2 = \sigma_{jj}$. 
\end{proposition}
The proof is given in Appendix \ref{proof:propmoments}. Two desirable properties can be observed there.  First, thanks to the   random term $Z_{ij}$, this distribution is 'over-dispersed' in the sense that the   variance exceed the  expectation: $\Var(Y_{ij}) > \Esp(Y_{ij})$. Second, the covariance $\Cov(Y_{ij}, Y_{ik})$ between the counts $Y_{ij}$ and $Y_{ik}$   has the same sign as the corresponding entry $\sigma_{jk}$ of the covariance matrix $\Sigma$.


%



\subsection{Identifiability from partial observations} 

We need to  prove that the parameters  $\theta= (\beta,\gamma, C)$ 
define uniquely the distribution of the observations $Y^O$. 
Or equivalently, that, if $(\theta,\theta')$ are such that the likelihood functions are equal, i.e.   $p_{\theta}(Y^O) = p_{\theta'}(Y^O)$ 
for any $Y^O$, then $\theta = \theta'$.
The identifiability relies on assumptions on the set of observed counts $\Ocal$ and on  $X^O$ the matrix of covariates  for the observed counts

First, note that  the covariance matrix $\Sigma$ is defined as $\Sigma = CC^\top$. As a consequence, $C$ is only identifiable up to a rotation matrix, so the identifiability can only be obtained for $\Sigma$. 
We present two results of identifiability of $(\beta, \gamma, \Sigma)$. In Proposition \ref{prop:identif1} we establish the identifiability of $(\beta, \gamma, \Sigma)$ with no assumption on the rank of $\Sigma$. In Proposition \ref{prop:identif2} we relax the assumptions on $\Ocal$ to obtain the identifiability of the extra diagonal terms of $\Sigma$ when $\Sigma$ is a low rank matrix.  These results were established by \citet{BCG24} for a full observation of the table: we extend them to partial observations and for a low rank covariance matrix $\Sigma$.  


Let us define 
$$
\Qcal =  \{(j,k)\in \set{p}^2 : \exists i  \in  \set{n}\text{ such that } \Omega_{ij} = \Omega_{ik}=1 \}. $$
$\Qcal$ contains the pairs or years (columns) $(j,k)$ such that   at least one site $i$ has been visited at years $j$ and $k$. 

\begin{proposition}[Identifiability of $(\beta, \gamma, \Sigma)$]\label{prop:identif1}
Assume that 
\begin{description}
  \item[\mbox{[A1]}]~$\Qcal = \set{p}^2$.  
  \item[\mbox{[A2]}]~$X^O$ is full rank
\end{description}
then $(\beta, \gamma, \Sigma)$ are identifiable. 
\end{proposition}


Assumption [A1] means that, for each pair of years $(j,k)$, at least one site $i$ has been observed both years. [A2] is a standard assumption for (generalized) linear models.
Both assumptions are straightforward to verify and are quite realistic in the context of ecological observations. 


\begin{proof}
The proof relies on the computation of the  three first moments of the \ZIPLN distribution.  

Let us denote $\mu_{ij}= x_{ij}^\top \beta$.  We first remind that the $k$-th moment of the log-normal distribution with parameters $(\mu,\sigma^2)$ is $\exp(k \mu + k^2 \sigma^2/2)$. Thus, using the first three moments of the Poisson distribution and Proposition \ref{prop:moments}, we can write that: 
\[
\begin{aligned}
\mathbb{E}[Y_{ij}] &= \pi_{ij} \cdot \exp(\mu_{ij} + \tfrac{1}{2} \sigma_j^2) , \\
\mathbb{E}[Y_{ij}^2] &= \pi_{ij} \cdot \left[ \exp(\mu_{ij} + \tfrac{1}{2} \sigma_j^2) + \exp(2\mu_{ij} + 2\sigma_j^2) \right] , \\
\mathbb{E}[Y_{ij}^3] &= \pi_{ij} \cdot \left[ \exp(\mu_{ij} + \tfrac{1}{2} \sigma_j^2) + 3\exp(2\mu_{ij} + 2\sigma_j^2) + \exp(3\mu_{ij} + \tfrac{9}{2}\sigma_j^2) \right] . 
\end{aligned}
\]
(the details are provided in Appendix section \ref{proof:propmoments}). 
Let us define the two following quantities: 
\begin{align*}
Q_1(Y_{ij}) 
&:= \frac{\mathbb{E}[Y_{ij}^2] }{\mathbb{E}[Y_{ij}]} - 1 =  \exp(\mu_{ij} + \tfrac{3}{2} \sigma_j^2) , \\
Q_2(Y_{ij}) 
&:= \mathbb{E}[Y_{ij}^3] - 3\mathbb{E}[Y_{ij}^2] + 2\mathbb{E}[Y_{ij}] 
= \pi_{ij} \exp(3\mu_{ij} + \tfrac{9}{2}\sigma_j^2)
= \pi_{ij} Q_1(Y_{ij})^3. 
\end{align*}
We are now able to express $\pi_{ij}$, $\mu_{ij}$ and $\sigma_j^2$ as functions of theses statistics and the expectation.  
\begin{align*}
  \pi_{ij} = & \frac{Q_2(Y_{ij})}{Q_1(Y_{ij})^3}, &
  \sigma_j^2 = &  \log \left(\frac{Q_2(Y_{ij})}{Q_1(Y_{ij})^2 \Esp(Y_{ij})}\right), &
  \mu_{ij} = & 4  \log \left( Q_1(Y_{ij})\right) + \frac{3}{2} \log \left( \frac{\Esp(Y_{ij})}{Q_2(Y_{ij})} \right).
\end{align*}
Using the fact that for all $(i,j) \in \Ocal$, $\mu_{ij} = x_{ij}^\top \beta$ and $\logit(\pi_{ij})  = x_{ij}^\top \gamma$, $\beta$ and $\gamma$ are identifiable provided $X^O$ is of full rank (Assumption [A2]). 
Finally, the covariance expression of Proposition \ref{prop:moments}, provides
$$
\sigma_{jk} = \log\left(\frac{ \Cov(Y_{ij}, Y_{ik})}{\Esp(Y_{ij}) \Esp(Y_{ik})} + 1\right) \qquad \text{for } (j,k) \mbox{ such that }  \exists \;  i \;  | \; (i,j) \in \Ocal \mbox{ and }  (i,k) \in \Ocal  
$$
And $ \sigma_{jk}$ is identifiable as soon as 
at least one site $i$ was visited during the two years $j$ and $k$.  
As a consequence all the entries of $\Sigma = (\sigma_{jk})_{(j,k) \in \set{p}}$ are identifiable 
under Assumption [A1].

\end{proof}


However, Assumptions [A1] and [A2] do not account for the fact that $\Sigma$ has rank $q < p$, which may require weaker assumptions.
Indeed, when $ \Sigma$ is assumed to be of low rank ($q \ll p$), one can hope to relax   assumption  [A1] to obtain the identifiability of the extra-diagonal terms of $\Sigma$.    
To go further, we use the results by \citet{BishopYu14} which provide conditions on $\Ocal$ for $\Sigma$ to be fully recovered from a partial identifiability, i.e. in cases where $\Qcal  \subsetneq \set{p}^2$. 

$\Qcal$ contains the pairs $(j,k)$ such that $\sigma_{jk}$ is identifiable from Proposition \ref{prop:identif1}, i.e.  at least one site $i$ has been visited at years $j$ and $k$. 
Observe that, if $\Qcal = \set{p}^2$, the complete identifiability of $\Sigma$ holds thanks to Proposition \ref{prop:identif1}: we therefore focus on the case where $\Qcal   \subsetneq \set{p}^2$. 
We provide conditions on $\Qcal$  and $\Sigma$ leading to the full recovering of $\Sigma$ from $(\sigma_{jk})_{(j,k) \in \Qcal}$. 

\begin{proposition}[Identifiability of $(\beta, \gamma, \Sigma)$ when $\Sigma$ has rank $q < p$]\label{prop:identif2} 
Assume that 
\begin{description}
 \item[\mbox{[A3]}] ~$\forall j  \in \set{p}$,  $ \exists i \in \set{n}$ such that $\Omega_{ij} = 1$.  
 \item[\mbox{[A4]}] ~$\Qcal$   can be written as the union of $r$ ordered product subsets   $\Qcal = \cup_{\ell=1}^r \Jcal_\ell \times \Jcal_\ell$   such that   
 $$
 \text{rank}\left((\sigma_{jk})_{(j,k) \in \varsigma_\ell^2}\right) \geq q
 \quad \mbox{  where } \quad  
\varsigma_\ell   =   \Jcal_{\ell} \cap (\cup_{\ell'=1}^{\ell-1} \Jcal_{\ell'}).  $$
\end{description}
Under Assumptions [A3] and [A4], $\Sigma$ can be completely recovered from $(\sigma_{jk})_{(j,k) \in \Qcal}$. As a consequence,  $(\beta,\gamma,\Sigma)$ are identifiable provided [A2], [A3] and [A4] hold.  
\end{proposition}

Assumption [A3] means that for any year $j$, at least one site $i$ has been visited.  Assumption [A4] imposes a structure on the pairs of years that are observed at least on one site. Such a structure is illustrated in Figure \ref{fig:Identif}. 

\begin{proof}
The proof follows from Theorem 2 of \cite{BishopYu14}. 
From Assumption [A3], we deduce that every $\sigma_j^2$ is identifiable. As a consequence, for every $j\in \set{p}$, $(j,j) \in \Qcal$ and so Assumption  (A1)  from \cite{BishopYu14} holds.  Our assumption [A4] gathers assumptions (A2), (A3) and (A4) from \cite{BishopYu14}. As a consequence Theorem 2 of \cite{BishopYu14} applies and  $\Sigma$ can be completely recovered from its principal sub-matrices  $(\Sigma_{jk})_{(j,k) \in \Qcal_\ell}$. 
\end{proof}

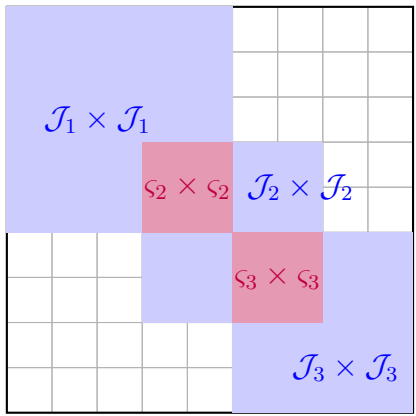
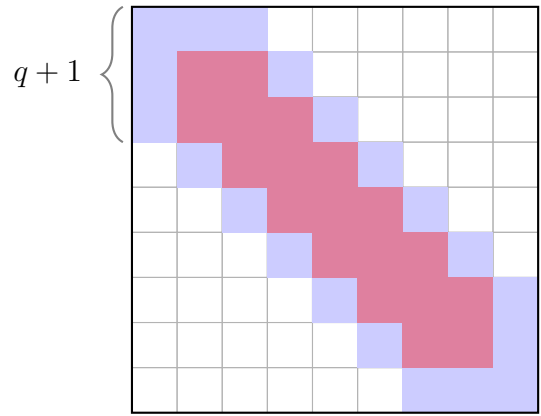
\begin{figure}[H]
\centering

\subfloat[$\Sigma$ of rank $2$ can be completely recovered from $\Sigma_{\Qcal}$ provided $\Sigma_{\varsigma_2\times\varsigma_2}$ and $\Sigma_{\varsigma_3\times\varsigma_3}$ are of rank $2$.\label{fig:Identif}]{
\centering
\begin{tikzpicture}[scale=0.6]

\def\n{9}      
\def\p{5}      
\def\q{4}      
\def\r{2}      
\def\s{4}      
\def\c{6}      

\pgfmathsetmacro{\a}{\p-\r+1}       
\pgfmathsetmacro{\b}{\a+\q-1}       
\pgfmathsetmacro{\d}{\c+\s-1}       

\pgfmathsetmacro{\aa}{6}

\foreach \i in {1,...,\n}{
  \foreach \j in {1,...,\n}{
    \draw[gray!60] (\i, -\j) rectangle (\i+1, -\j+1);
  }
}

\draw[thick] (1, 0) rectangle (\n+1, -\n);

\foreach \i in {1,...,\p}{
  \foreach \j in {1,...,\p}{
    \fill[blue!20] (\i, -\j) rectangle (\i+1, -\j+1);
  }
}

\foreach \i in {\a,...,\b}{
  \foreach \j in {\a,...,\b}{
    \fill[blue!20] (\i, -\j) rectangle (\i+1, -\j+1);
  }
}

\foreach \i in {\c,...,\d}{
  \foreach \j in {\c,...,\d}{
    \fill[blue!20] (\i, -\j) rectangle (\i+1, -\j+1);
  }
}

\foreach \i in {\a,...,\p}{
  \foreach \j in {\a,...,\p}{
    \fill[purple!40] (\i, -\j) rectangle (\i+1, -\j+1);
  }
}

\foreach \i in {\aa,...,\b}{
  \foreach \j in {\aa,...,\b}{
    \fill[purple!40] (\i, -\j) rectangle (\i+1, -\j+1);
  }
}

\node[blue] at (3,-\p/2) {$\Jcal_1 \times \Jcal_1$};
\node[purple] at (\a+1,-\p+1) {$\varsigma_2 \times \varsigma_2$};
\node[purple] at (\aa+1,-\b/2-1/2-2) {$\varsigma_3 \times \varsigma_3$};
\node[blue] at (7+1/2,-4) {$\Jcal_2 \times \Jcal_2$};
\node[blue] at (7+3/2,-8) {$\Jcal_3 \times \Jcal_3$};

\end{tikzpicture}
}
\hfill
\subfloat[Illustration of a small $\Qcal$ with respect to $p^2$.]{%
\centering
\label{fig:identif2}
\begin{tikzpicture}[scale=0.6]

\def\p{9}     
\def\q{2}     

\foreach \i in {1,...,\p}{
  \foreach \j in {1,...,\p}{
    \draw[gray!60] (\i, -\j) rectangle (\i+1, -\j+1);
  }
}

\foreach \k in {1,...,\numexpr\p-\q\relax}{
  \foreach \i in {\k,...,\numexpr\k+\q\relax}{
    \foreach \j in {\k,...,\numexpr\k+\q\relax}{
      \fill[blue!20] (\i, -\j) rectangle (\i+1, -\j+1);
    }
  }
}

\foreach \k in {1,...,\numexpr\p-\q-1\relax}{
  \foreach \i in {\numexpr\k+1\relax,...,\numexpr\k+\q\relax}{
    \foreach \j in {\numexpr\k+1\relax,...,\numexpr\k+\q\relax}{
      \fill[purple!50] (\i, -\j) rectangle (\i+1, -\j+1);
    }
  }
}

\draw[thick] (1,0) rectangle (\p+1,-\p);

\draw[decorate, decoration={brace, amplitude=8pt}, thick, color=gray]
 (0.8,-\q-1) -- (0.8,0)
 node[midway, xshift=-1cm, color=black] {$q+1$};

\end{tikzpicture}
}

\caption{Illustrations of $\Qcal$ defined in Proposition \ref{prop:identif2} for $p=9$ and $q=2$.}
\label{fig:model_outputs_north_africa}
\end{figure}

\paragraph{Illustration.} In order to better understand the implications of Proposition \ref{prop:identif2}, we present in Figure \ref{fig:identif2} a configuration where  $\Qcal$ is small with respect to $\set{p}^2$. 
Assume that  $\Sigma$ is of rank $q$ and that the set $\Ocal$ is such that $\Qcal$ is defined as a union of $p-q$ matrices $\Qcal_\ell$ of size $(q+1)\times (q+1)$ with intersection  of size $q \times q$, organized as in Figure \ref{fig:identif2}. In that case, $\Qcal$ is of cardinal $p^2- (p-q-1)(p-q) = q^2 + p(2q-1) \leq p^2$. As a consequence, taking into account the rank of $\Sigma$ allows one to drastically reduce the set number of observations. 

\begin{figure}
\centering 

\label{}
\end{figure}
\section{Inference}\label{ref:inference}

We introduce the method we use for the inference of the parameter set $\theta = (\beta, \gamma, C)$, based on the available data $X^O$, $Y^O$. We denote by $U = [U_{ij}]_{1 \leq i \leq n, 1\leq j\leq p}$ the matrix of indicator variables and by $U^O$ the set of indicator variables restricted to the observed sites and years: $U^O = \{U_{ij}: (i, j) \in \Ocal\}$.

\subsection{Lower bound of the likelihood}

\paragraph{Likelihood.} 
The usual way to go with latent variable models is the Expectation-Maximization (EM) algorithm \cite{dempster1977maximum}, which relies on the following decomposition of the log-likelihood of the observed data:
\begin{align*}
  \log p_\theta (Y^O)
  & = \Esp_\theta [\log p_\theta (Y^O, U, W) \mid Y^O] - \Esp_\theta [\log p_\theta (U, W \mid Y^O) \mid Y^O] \\
  & = \Esp_\theta [\log p_\theta (Y^O, U^O, W) \mid Y^O] - \Esp_\theta [\log p_\theta (U^O, W \mid Y^O) \mid Y^O] 
\end{align*}
because of the independence of the indicator variables $(U_{ij})$.
$\log p_\theta (Y^O, U^O,W)$ is often called the complete (log-)likelihood and $ -\Esp_\theta [\log p_\theta (U^O, W \mid Y^O) \mid Y^O]$ is the conditional entropy of $(U^O, W)$ given $Y^O$, which we denote by $\Hcal[p_\theta (U^O, W \mid Y^O)]$. 
The interest of this decomposition, is that it involves the so-called complete log-likelihood $\log p_\theta (Y^O, U^O, W)$, which is usually easier to deal with. Because the sites are independent, the complete log-likelihood of the observed counts under Model \ref{mod:ZIPLNPCA} is
\begin{equation} \label{eq:logLik}
  \log p_\theta(Y^O, U^O, W) = \sum_{i=1}^n \log p_\theta(Y_i^O, U_i^O, W_i)
\end{equation}
where
\begin{align} \label{eq:logLiki}
  \log p_\theta(Y_i^O, U_i^O, W_i)
  & = \log p_\theta (Y_i^O \mid U_i^O, W_i) + \log p_\theta(U_i^O) + \log p_\theta(W_i) \nonumber \\
  & = \sum_{j=1}^p \left(\Omega_{ij} \log p_\theta (Y_{ij} \mid U_{ij}, W_i) +  \Omega_{ij} \log p_\theta(U_{ij})\right) + \log p_\theta(W_i) \nonumber \\
  & = \sum_{j=1}^p \Omega_{ij} U_{ij} \left(Y_{ij}(x_{ij}^\top \beta + W_i^\top C_j) - \exp(x_{ij}^\top \beta + W_i^\top C_j)\right) \\
  & \quad + \sum_{j=1}^p \Omega_{ij} \left(U_{ij} (x_{ij}^\top \gamma) - \log(1 + \exp(x_{ij}^\top \gamma))\right) 
  - \frac12 \|W_i\|^2 \nonumber\\
  & \quad - \frac{q}2 \log{2 \pi} - \sum_{j=1}^p \Omega_{ij} U_{ij} \log(Y_{ij}!) \nonumber
\end{align}
where $C_j$ ($1 \leq j \leq p$) stands for the vector made of $j$-th row of $C$. Observe that the terms from the last line do not depend on $\theta$

The counterpart of the EM algorithm is that it requires the calculation of a series of moments under the conditional distribution $p_\theta ( U^O, W \mid Y^O)$. Unfortunately, as in the case of Poisson log-normal model, this conditional distribution has no close form and the calculation of its moments is intractable  \cite{chiquet2018variational}. To circumvent this issue, we resort to a variational approximation.

\paragraph{Variational Approximation.} 
Variational approximation refers to a wide variety of techniques, which consists in carrying approximate inference for latent variable models, replacing the intractable true conditional distribution $p_\theta (U^O, W \mid Y^O)$ with some approximate distribution $\pt_\psi(U^O, W)$, chosen in a convenient parametric distribution class \cite{WaJ08} ($\psi$ being the parameter ruling the approximate distribution). The estimation procedure then consists in maximizing the evidence lower bound (ELBO)
\begin{align} \label{eq:defElbo}
     J(\theta, \psi)
     & = \log p_\theta(Y^O) - \text{KL}[\pt_\psi(U^O, W) \ \mid p_\theta(U^O, W \mid Y^O)] \\
     & = \Espt_\psi[\log p_\theta(Y^O, U^O, W)] + \Hcal(\pt_\psi(U^O, W)) \nonumber
\end{align}
where $\text{KL}$ stands for the K\"ullback-Leibler divergence and $\Espt_\psi$ for the expectation under distribution $\pt_\psi$. We remind that $\pt_\psi(Z)$ is supposed to be close to $p_\theta(Z \mid Y^O)$, so $\Espt[f(Z)] \simeq \Esp_\theta[Y \mid Y^O]$. Because the $\text{KL}$ divergence is non-negative, the ELBO is obviously a lower bound of the log-likelihood of the observed data: $J(\theta, \psi) \leq \log p_\theta(Y^O)$. The $\psi$ of the approximate distribution $\pt$, is called the variational parameter.

Because of the independence of the sites, the intractable conditional distribution has obviously a product form:
$$
p_\theta (W, U^O \mid Y^O)
= \prod_{i=1}^n p_\theta (W_i, U_i^O \mid Y_i^O),
$$
so we may approximate each site-specific conditional distribution $p_\theta (W_i, U_i^O \mid Y_i^O)$ independently. Following \citet{chiquet2018variational} and \citet{BCG24}, we seek for the approximate distribution within the set of products of $|\Ocal_i|$ independent Bernoulli (for the indicator variables $U_{ij}$) with a $q$-dimensional Gaussian (for $W_i$):
\begin{equation} \label{eq:apprClass}
  \pt_{\psi_i} = \left(\bigotimes_{j \in \Ocal_i} \Bcal(\xi_{ij})\right) \otimes \Ncal(m_i, S_i),
\end{equation}
where each mean vector $m_i$ has dimensional $q$ and each $q \times q$ variance matrix $S_i$ is supposed to be diagonal: $S_i = \diag(s_i)$, where $s_i$ stands for the variance vector $s_i = [s^2_{i1} \dots s^2_{iq}]$. Hence, the variational parameter is $\psi_i = ((\xi_{ij})_{j \in \Ocal_i}, m_i, s_i)$. We stack all the mean vectors to form the $n \times q$ matrix $M = [m_1 \dots m_n]^\top$ and the variances vectors $s_i$ to form the matrix $S = [s_i]^\top_{1 \leq i \leq n}$. We further define $\xi^O = \{\xi_{ij}: (i, j) \in \Ocal\}$ so that the whole set of variational parameters is $\psi = (\xi^O, M, S)$.

\subsection{Variational inference}

Based on the complete likelihood given in Equations \eqref{eq:logLik} and \eqref{eq:logLiki}, we may write the ELBO defined in Equation \eqref{eq:defElbo} as
\begin{align} \label{eq:elbo}
    J(\theta, \psi) = \sum_{i=1}^n J_i(\theta, \psi)
\end{align} 
where $J_i(\theta, \psi)$ is the contribution of site $i$ to the ELBO:
\begin{align} \label{eq:elboi}
  J_i(\theta, \psi)
  & := \Espt_\psi\left[\log p_\theta(Y_i^O \mid U_i^O, W_i) + \log p_\theta(U_i^O) + \log p_\theta(W_i) - \log \pt_\psi(U_i^O, W_i)\right] \nonumber \\
  & = \Espt_\psi\left[\sum_{j=1}^p \Omega_{ij} \left(U_{ij} \log p_\theta (Y_{ij} \mid U_{ij}, W_i) + \log p_\theta(U_{ij}) - \log \pt_\psi(U_{ij}) \right) \right] \nonumber \\
  & \quad + \Espt_\psi\left[\log p_\theta(W_i) - \log \pt_\psi(W_i) \right]  \nonumber\\
  & = \sum_{j=1}^p \Omega_{ij} \xi_{ij} \left(Y_{ij}(x_{ij}^\top \beta + m_i^\top C_j) - \exp(x_{ij}^\top \beta + m_i^\top C_j + C_j^\top S_i C_j/2)\right) \\
  & \quad + \sum_{j=1}^p \Omega_{ij} \left(\xi_{ij} (x_{ij}^\top \gamma) - \log(1 + \exp(x_{ij}^\top \gamma) - \xi_{ij} \log(\xi_{ij})) - (1 - \xi_{ij}) \log(1 - \xi_{ij})\right) \nonumber \\
  & \quad - \frac12 \left(\|m_i\|^2 + \text{tr}(S_i) - \sum_{k=1}^q \log(s_{ik}^2) \right) 
  - \sum_{j=1}^p \Omega_{ij} \xi_{ij} \log(Y_{ij}!). \nonumber
\end{align}

\paragraph{Maximization of the ELBO.}
Variational approximations are often used in the context of the so-called variational EM (VEM) algorithm, which mimics the classical EM algorithm \cite{WaJ08}, alternating the optimization with respect to the model parameters $\theta$ (M step) and with respect to the variational parameters $\psi$ (VE step). In the present case, we choose to maximize the ELBO $J(\theta, \psi)$ with respect to the model parameters $\theta = (\beta, \gamma, C)$ and the variational parameters $\psi = (\xi, M, S)$ at once, using gradient ascent. Observe that, because of missing data, the ELBO to be maximised differs from this of \citet{BCG24}.

Denoting by $A_{ij}$ the expectation of the abundance $Y_{ij}$ conditional on the presence of the species under the variational approximation
$$
A_{ij} 
:= \Espt_\psi[Y_{ij} \mid U_{ij} = 1]
= \exp(x_{ij}^\top \beta + m_i^\top C_j + C_j^\top S_i C_j/2),
$$
and by $\odot$ the Hadamard product, the derivatives of the contribution $J_i(\theta, \psi)$ of the $i$-th site to the ELBO with respect to the model parameters are
\begin{align*}
  \partial_\gamma J_i & = \sum_j \Omega_{ij} (\xi_{ij} - \pi_{ij} ) x_{ij}, & 
  \partial_{C_j} J_i & = \Omega_{ij}\xi_{ij} ((Y_{ij} - A_{ij})m_i - A_{ij} (C_j \odot s_i)), \\
  \partial_\beta J_i & = \sum_j \Omega_{ij} \xi_{ij} (Y_{ij} - A_{ij}) x_{ij}.
\end{align*}
As for the variational parameters, the derivative of $J_i(\theta, \psi)$ are non zero only for the elements of $\psi_i$ and are
\begin{align*}
  \partial_{m_i} J_i & = \sum_j \Omega_{ij} \xi_{ij} (Y_{ij} - A_{ij})C_j - m_i, \\
  \partial_{s_i} J_i & = \frac12 \left(s_i^{\odot -1} - \One_q - \sum_j \Omega_{ij} \xi_{ij} A_{ij} (C_j \odot C_j)\right), \\
\partial_{\xi_{ij}} J_i & = \nu_{ij} + \Omega_{ij} \left(-A_{ij} + Y_{ij} (\mu_{ij} + C_j^\top m_i) - \log(Y_{ij} !)\right) - \log \left(\frac{\xi_{ij}}{1 - \xi_{ij}}\right),
\end{align*}
where $s_i^{\odot -1} = [1/s_{i1}^2 \dots 1/s_{iq}^2]$ and $1_q$ is the $q$-dimensional vector filled with ones.

The VEM algorithm consists of iteratively updating $\theta$ and $\psi$ as
$\theta^{(h+1)} = \argmax_\theta J(\theta, \psi^{(h)})$ and $\psi^{(h+1)} = \argmax_\psi J(\theta^{(h+1)}, \psi)$, 
until convergence. 

\subsection{Selection of the latent dimension $q$} 

The dimension of the latent space $q$ has to be chosen from the data. Because the \ZIPLNPCA model is identifiable and regular, a Bayesian Information Criterion (BIC) criterion can be formally derived as in \citet{Sch78}. The same argument holds to derive an Integrated Completed likelihood (ICL) criterion as in \citet{BCG00}. The resulting criteria are
\begin{align} 
    \mbox{BIC}(Y^O,q)& = \log p_{\widehat{\theta}_q}(Y^O) - \frac{1}{2} \mbox{pen}(d,q)\log (n)  \label{eq:BIC} \\  
    \mbox{ICL}(Y^O,q) &= \log p_{\widehat{\theta}_q}(Y^O)  - \Hcal(p_{\widehat{\theta}}(U^O, W \mid Y^O))  - \frac{1}{2} \mbox{pen}(d,q)\log (n) \label{eq:ICL}
\end{align}
where the penalty term 
$$
\mbox{pen}(d,q) = (pq + 2d)
$$ 
takes into account the dimensions of $\gamma$, $\beta$ and $C$. 
The difference between BIC and ICL lies in the $ \Hcal(p_{\widehat{\theta}}(U^O, W \mid Y^O)) $ term. With respect to BIC, ICL will penalize the entropy of the latent distribution, which is directly linked to the latent space dimension $q$. On the contrary, BIC will allow larger $q$, which will catch a more complex dependence between times. \\ 
In practice, because the log-likelihood $\log p_{\widehat{\theta}_q}(Y^O)$ has no explicit expression under the \ZIPLNPCA model, we replace it with the ELBO $J(\widehat{\theta}, \widehat{\psi})$ which is assumed to be a good approximation of it, provided that the divergence between $p_{\widehat{\theta}}(U^O, W~\mid~Y^O)$ and $\pt_{\widehat{\psi}}(U^O, W)$ is small. Similarly, in the ICL criterion, the entropy term $\Hcal(p_{\widehat{\theta}}(U^O, W~\mid~Y^O))$ is replaced by  $\Hcal(\pt_{\widehat{\psi}}(U^O, W) )$ which  can be computed explicitly as : $\sum_{i=1}^n \sum_{j=1}^p - (\xi_{ij} \log(\xi_{ij}) + (1 - \xi_{ij}) \log(1 - \xi_{ij})) + \frac{1}{2} \sum_{i = 1}^n \sum_{j = 1}^p \log(s_{ij}) + \frac{nq}{2}$ if $0 < \xi_{ij} < 1$ and $\frac{1}{2} \sum_{i = 1}^n \sum_{j = 1}^p \log(s_{ij}) + \frac{nq}{2}$ ortherwise.


\subsection{Confidence intervals for the model parameters} \label{sec:varTheta}


We now derive the approximate variance of the parameters estimator $\widehat{\theta}$, which results from the variational algorithm presented above. 
To this aim, we follow the line introduced by \citet{WM19}, which borrows from the theory of $M$-estimators. 
Denoting 
$$
\widetilde{\psi}(\theta) := \arg\max_\psi J(\theta, \psi),
$$
the variational estimator $\widehat{\theta}$ is the $M$-estimator associated with the contrast
$$
\widetilde{J}(\theta) := J(\theta, \widetilde{\psi}(\theta)).
$$
Provided that $\widehat{\theta}$ is consistent, the theory of $M$-estimation then yields that its asymptotic variance is $\Var(\theta) = C(\theta)^{-1} D(\theta) C(\theta)^{-1}$ (see, e.g., \cite{VdV2000}, chapter 5), with
\begin{align*}
    C(\theta) & = \Esp_\theta\left[\nabla^2_{\theta \theta}\widetilde{J}(\theta)\right], & 
    D(\theta) & = \Esp_\theta\left[\nabla_\theta \widetilde{J}(\theta) [\nabla_\theta \widetilde{J}(\theta)]^\top\right].
\end{align*}
The derivatives of $\widetilde{J}$ with respect to $\theta$ can then be obtained \citep[see][]{WM19} as:
\begin{align*}
    \nabla_{\theta} \widetilde{J}(\theta) 
    & = \nabla_{\theta} J(\theta, \psi), \\
    \nabla^2_{\theta \theta} \widetilde{J}(\theta) 
    & = \nabla_{\theta \theta} J(\theta, \psi) - \nabla_{\theta \psi} J(\theta, \psi) \left[\nabla_{\psi \psi} J(\theta, \psi)\right]^{-1} [\nabla_{\theta \psi} J(\theta, \psi)]^\top
\end{align*}
evaluated at $\psi = \widetilde{\psi}(\theta)$ (which makes the first equation hold because, $\widetilde{\psi}(\theta)$ being a maximum, $\nabla_{\psi} J(\theta, \psi) = 0$ for $\psi = \widetilde{\psi}(\theta)$).
Finally, unbiased estimates of $C(\theta)$ and $D(\theta)$ are given by the following empirical means
\begin{align*}
    \widehat{C}_n 
    & = \frac{1}{n} \sum_{i = 1}^n \nabla_{\theta \theta} J_i(\widehat{\theta}, \widetilde{\psi}(\widehat{\theta})) - \nabla_{\theta \psi} J_i(\widehat{\theta}, \widetilde{\psi}(\widehat{\theta})) \left[\nabla_{\psi \psi} J_i(\widehat{\theta}, \widetilde{\psi}(\widehat{\theta}))\right]^{-1} [\nabla_{\theta \psi} J_i(\widehat{\theta}, \widetilde{\psi}(\widehat{\theta}))]^\top, \\
    \widehat{D}_n 
    & = \frac{1}{n} \sum_{i = 1}^n \nabla_\theta J_i(\widehat{\theta}, \widetilde{\psi}(\widehat{\theta})) \; [\nabla_\theta J_i(\widehat{\theta}, \widetilde{\psi}(\widehat{\theta}))]^\top.
\end{align*}
Approximate confidence intervals for the model parameters can then be constructed, based on the estimated asymptotic variance $\widehat{\Var}_n(\widehat{\theta}) = \widehat{C}_n^{-1} \widehat{D}_n \widehat{C}_n^{-1}$ and the Gaussian approximation 
\begin{equation} \label{eq:asympNormal}
\widehat{\theta} \approx \Ncal\left(\theta^*, \widehat{\Var}_n(\widehat{\theta}) \right),
\end{equation}
$\theta^*$ being the true value of the parameter.
Observe that the validity of the estimated variance and the Gaussian approximation (also used in \cite{BCM24}) relies on the consistency of $\widehat{\theta}$, which we do not prove here.
As a consequence, the variance of the estimates we provide is only approximate. Its accuracy will be assessed in Section \ref{sec:data}, through the accuracy of the prediction intervals.

\subsection{Imputation}

We now turn to the imputation part, that is to the prediction of the unobserved counts $Y^{M}$  given the observed ones $Y^O$ (and the covariates). A natural predictor for the missing observation $Y_{ij}$ with $j \in \Ocal_i$ is its conditional expectation given the corresponding vector of covariates $x_{ij}$ and the set of observed counts $Y^O$: $\Esp_\theta[Y_{ij} \mid Y^O]$. Because of the independence between the sites, conditioning on the observed counts within the same site (that is: $Y_i^O$) is obviously sufficient, so we aim at evaluating $\Esp_\theta[Y_{ij} \mid Y_i^O]$. \\
Unfortunately, we do not have access to this conditional distribution $p_\theta(\cdot \mid Y^O)$, but only to its variational approximation $q_\psi(\cdot)$. Under this approximation, the distribution of the missing observation $Y_{ij}$ is a univariate \ZIPLN distribution with presence probability $\xi_{ij}$, log-expectation $x_{ij}^\top \beta + m_{ij}$ and latent variance $C_j^\top S_i C_j$. Therefore, following Proposition \ref{prop:moments}, we consider the conditional prediction: 
\begin{align} \label{eq:condPred}
    \widetilde{Y}_{ij}
    = \widetilde{\xi}_{ij} \exp\left(x_{ij}^\top \widehat{\beta} + \widehat{C}_j^\top \widetilde{m}_i + \widehat{C}_j^\top \widetilde{S}_i \widehat{C}_j  / 2 \right),
\end{align}
where $\widetilde{\xi}_{ij}$, $\widetilde{m}_i$ and $\widetilde{S}_i$ are the variational parameter estimates. Observe that the latent correlation structure imposed by Model \ref{mod:ZIPLNPCA} enables us to borrow information from one year to another, within a same site. 

In addition, we aim to assign each prediction a prediction interval that reflects both the intrinsic variability of Model \ref{mod:ZIPLNPCA} and the uncertainty associated with parameter estimation. To obtain these intervals, we designed the Monte Carlo Algorithm \ref{algo:condPI}.

\begin{algorithm}[H]
\caption{Conditional predictions and prediction intervals}
\label{algo:condPI}
    \begin{description}
        \item[{\sl Input:}]~Estimates $\widehat{\theta}$ and $\widehat{\Var}(\widehat{\theta})$, number of particles $B$. 
        \item[{\sl Output:}]~Conditional intervals for the $\Esp(Y_{ij} \mid Y^O)$ and conditional prediction interval for the $Y_{ij}$. 
        \item[{\sl Algorithm:}] ~
        \begin{enumerate}
            \item For $b = 1 \dots B$, 
            \begin{enumerate}
                \item Sample $\theta^b \sim \Ncal\left(\widehat{\theta}, \widehat{\Var}(\widehat{\theta})\right)$, and denote $\gamma^b$, $\beta^b$ and $\Sigma^b$ the resulting model parameters;
                \item Perform a VE-step to get $\widetilde{\psi}(\theta^b)$, and denote $\xi^b$, $M^b$ and $S^b$ the resulting variational parameters and, for each $(i, j) \in \Mcal$, compute
                $$
                \widetilde{Y}^b_{ij}
                = \widetilde{\xi}^b_{ij} \exp\left(x_{ij}^\top \beta^b + (m_i^b)^\top C^b_j + (C^b_j)^\top S^b_i C^b_j  / 2 \right),
                $$
                \item For each site $i$, such that $\Mcal_i \neq \varnothing$, sample independently $W^b_i \sim \Ncal(m^b_i, S^b_i)$, compute $Z_i^b = C^b W_i^b$, and for each year $j \in \Mcal_i$, sample $U^b_{ij} \sim \Bcal(\xi_{ij}^b)$, set $Y^b_{ij} = 0$ if $U^b_{ij} = 0$ and, otherwise, sample independently 
                $$
                Y^b_{ij} \sim \Pcal\left(\exp\left(x_{ij}^\top \beta^b + (mi^b)^\top C^b_j\right)\right);
                $$ 
            \end{enumerate}
            \item Determine the intervals for each missing  observation $(i, j) \in \Mcal$:
            \begin{itemize}
                \item Use the empirical quantiles of the $(\widetilde{Y}^b_{ij})_{1 \leq b \leq B}$ to determine a conditional confidence interval for each conditional expected count $\Esp(Y_{ij} \mid Y^O)$;
                \item Use the empirical quantiles of the $(Y^b_{ij})_{1 \leq b \leq B}$ to determine a conditional prediction interval for $Y_{ij}$.
            \end{itemize}
        \end{enumerate}
    \end{description}
\end{algorithm}

Algorithm \ref{algo:condPI} is computationally demanding because of the VE-step performed in step $(b)$ and because of the different sampling performed in steps $(b)$ and $(c)$. If only a confidence interval for $\Esp(Y_{ij} \mid Y^O)$ is needed, the computational burden can be reduced by skipping step $(c)$. Still the most demanding step $(b)$ remains, but it is required to share information between observed and unobserved counts within a same site.

\paragraph{Marginal prediction.}
Alternatively, we may consider the unconditional (that is marginal) expectation $\Esp_\theta(Y_{ij})$ and consider the prediction
\begin{align} \label{eq:margPred}
    \widehat{Y}_{ij}
    = \widehat{\pi}_{ij} \exp\left(x_{ij}^\top \widehat{\beta} + \widehat{\sigma}_{jj} / 2\right).
\end{align}
A Monte-Carlo algorithm can be designed in the same way to get both marginal confidence and prediction intervals, which does not resort to any additional VE-step: see Algorithm \ref{algo:margPI} in Appendix \ref{app:margPred}. If only a marginal confidence interval is needed, Algorithm \ref{algo:margPI} reduces to step $(a)$, which makes it even faster.


\section{Illustrations on waterbirds survey}\label{sec:data} 

\newcommand{\covar}{G}
\newcommand{\full}{F}

To  assess the performance and illustrate the use of our model and the associated inference strategy, we consider two waterbird population monitoring: one from Europe and one from North Africa. As the data used in this section is not entirely public, the bird species we will examine will be designated by capital letters (A, B, and C).

\subsection{European waterbird}


To assess the accuracy of our predictions in a realistic setting, we first consider a European monitoring, which is complete (i.e. which has no missing data), and in which we  artificially introduce missing values.

\paragraph{Dataset.}
The data set we are examining comprises winter censuses of species~A carried out at 228 sites in European countries between 1992 and 2012. 
The proportion of zeros in this data set is 37\% of all site-year observations. The average number of individuals counted is approximately 11, with a variance of 573, reflecting high overdispersion.
Three types of covariates (such as described in Section \ref{sec:notations}) are also available:  covariates on sites (latitude, distance to the city, etc.), covariates on years (weather) and covariates on sites/year couples (winter precipitation) (see Appendix~\ref{app:covariates} for all covariates).

\paragraph{Objectives.}
This section focuses on the quality of imputations within our model.
The first objective of this section is to evaluate the performances of the proposed methodology under different patterns of missing data. 
Secondly, we question the importance of the choice of covariates for imputation purposes.  
Thirdly, we aim to compare the performance of a zero-inflated model with that of a non–zero-inflated model.
Finally, we want to assess the covering rate of the prediction intervals.

\paragraph{Simulation design.} 
We considered four different scenarios for the missing data pattern : one Missing Completely at Random (MCAR) scenario and three Missing at Random (MAR) scenarios. 
In the MCAR case, missingness was introduced uniformly across sites and years. 
The MAR scenarios incorporated structured missingness: (i) depending on the year, with a decreasing proportion of missing values over time; (ii) depending on the site, with some sites consistently less frequently monitored; and (iii) depending jointly on sites and years, a configuration closely matching the monitoring conditions observed in North Africa. For each mechanism, we removed a proportion of  5\%, 30\%, 50\%, and 70\% of the observations from the original dataset.  100 replicates were generated for each configuration, resulting in $4 \times 4 \times 100 = 1600$ data sets.

We considered two alternative sets of covariates for the imputation step, denoted $X^{\covar}$ (for genuine) and $X^{\full}$ (for full). 
Both are constructed according to the general design matrix structure described in Section~2,
$$
X = \bigl[ X_R \otimes I_p \;\; I_n \otimes X_C \;\; X_E \bigr].
$$
The set $X^{\covar}$ includes year- and site-specific covariates through the matrices $X_R$ and $X_C$. 
By contrast, $X^{\full}$ includes only sites- and years-specific effects, with $X_R = I_n$ and $X_C = I_p$, and the first column of each matrix removed to ensure identifiability. The model based on the second set obviously includes more information about each site (or year) than the first, as the linear span of $X^{\covar}$ is strictly included in the linear span of $X^{\full}$.


\paragraph{Evaluation criteria.} 

We fitted the \ZIPLNPCA  model with the two sets of covariates and the \PLNPCA model with the second set of covariates to each of the generated data sets. 
In both cases, we selected $q$, the dimension of the latent space, using the BIC criterion for a single replicate per configuration (pattern $\times$ rate of missing data) and kept these value for $q$ for all other replicates.
For each missing observation $(i, j) \in \Mcal$, we computed the conditional prediction $\widetilde{Y}_{ij}$ as defined in Equation \eqref{eq:condPred} and the conditional prediction interval using Algorithm \ref{algo:condPI}.
Because the actual count $Y_{ij}$ is known, we measured the accuracy of the imputation with the absolute prediction error $|\widetilde{Y}_{ij} - Y_{ij}|$.
To limit the computation time, we evaluated the prediction intervals only for the first 10 replicates of each type. 



\paragraph{Results.} 
The proposed method turned out to be relatively insensitive to the missingness pattern. Therefore we only present the results for the MAR site and year scenarios, as they represent the patterns closest to reality, since they generate large spatio-temporal areas of missing data. 
The results under the three other missingness scenario are similar and given in Figure~\ref{fig:ratio_comparison} in Appendix~\ref{app:sim2}.
We designate by $(X, \PLNPCA)$ the results obtained with the covariate matrix $X$ ($X^\covar$ or $X^\full$) and the model, say, \PLNPCA.


The first aim was to evaluate the role of covariates in improving imputation accuracy.
The top left panel of Figure \ref{fig:compare} displays the ratio between the absolute prediction errors obtained with either the genuine data set or the set including the site and year effects. The figure show that the latter specification consistently yields smaller prediction errors across all missingness levels.
This result is expected, as site and year effects can be viewed as the most informative covariates describing the years and sites, but it emphasizes the dramatic influence of the available covariates on the accuracy of the predictions and motivates the effort to collect them.
The additional simulation study presented in Appendix~\ref{app:sim1} shows that imputation accuracy improves as the proportion of variance explained by the covariates increases, and the results presented here provide empirical support for  this relationship. Consequently, for imputation purposes, using site and year effects (that is $X^{\full}$) is more appropriate than relying on the covariates describing the sites and the years (that is $X^{\covar}$).

\begin{figure}[H]
    \centering
    \begin{tabular}{cc}
         Prediction error & Prediction error \\
         ($X^{\covar}$, \ZIPLNPCA) vs ($X^{\full}$, \ZIPLNPCA) &  
         ($X^{\full}$, \ZIPLNPCA) vs  ($X^{\full}$, \PLNPCA) \\
         \includegraphics[width=.45\textwidth, height=.15\textheight]{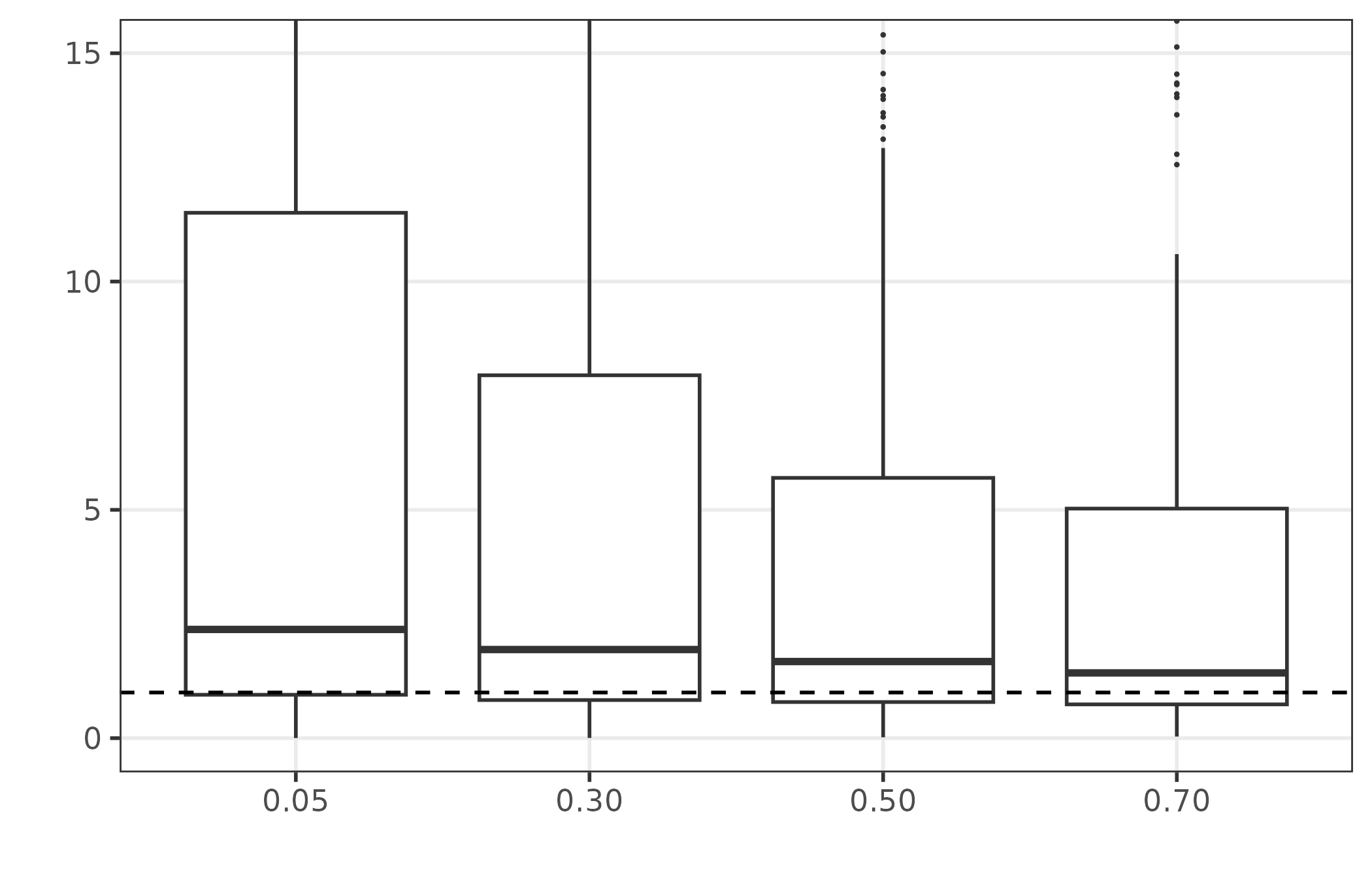} &
         \includegraphics[width=.45\textwidth, height=.15\textheight]{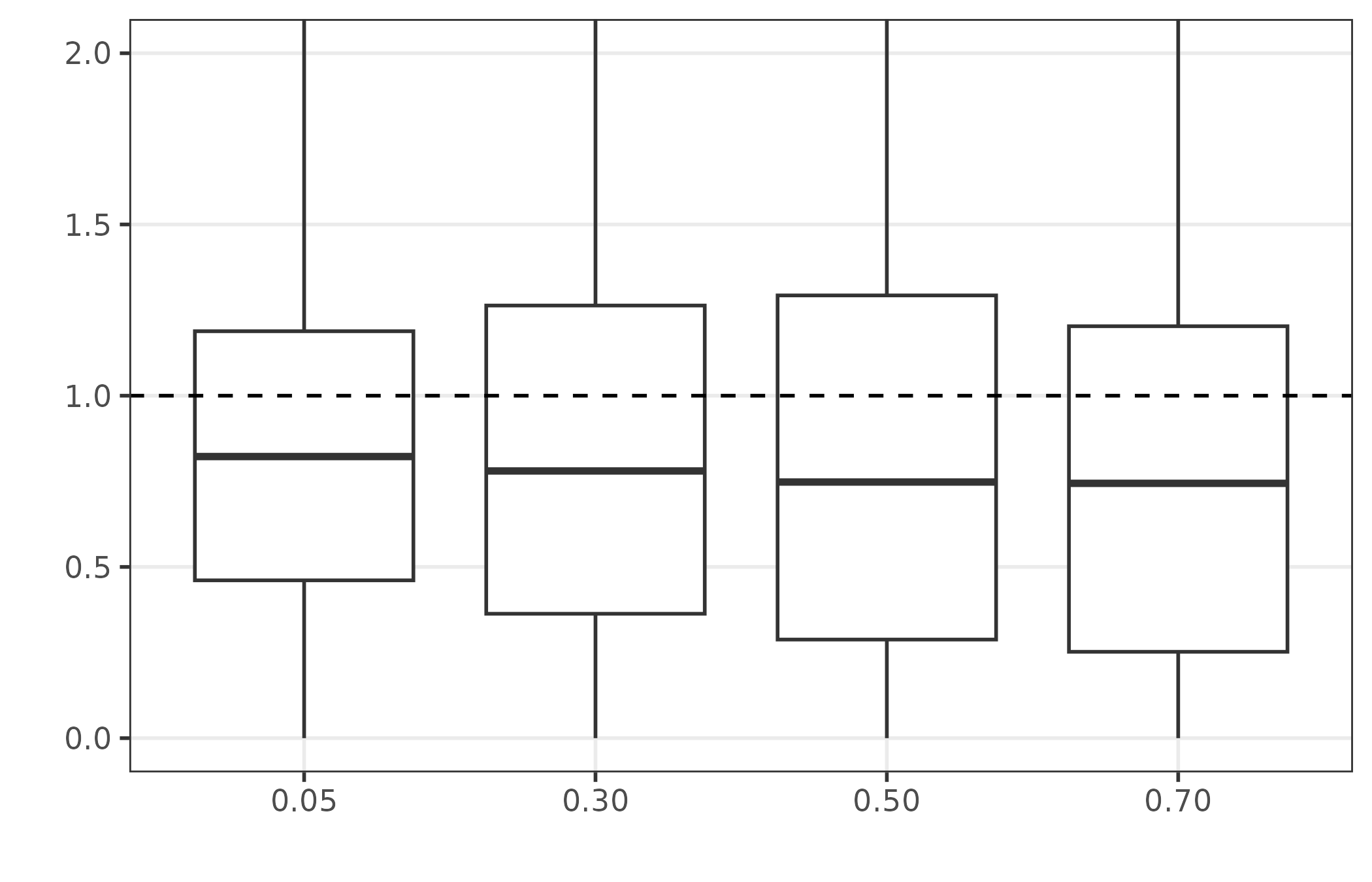} \\
         Width of the prediction intervals, MAR & Width of the prediction intervals, MCAR \\ 
         ($X^{\full}$, \ZIPLNPCA) & ($X^{\full}$, \ZIPLNPCA) \\
         \includegraphics[width=.45\textwidth, height=.15\textheight]{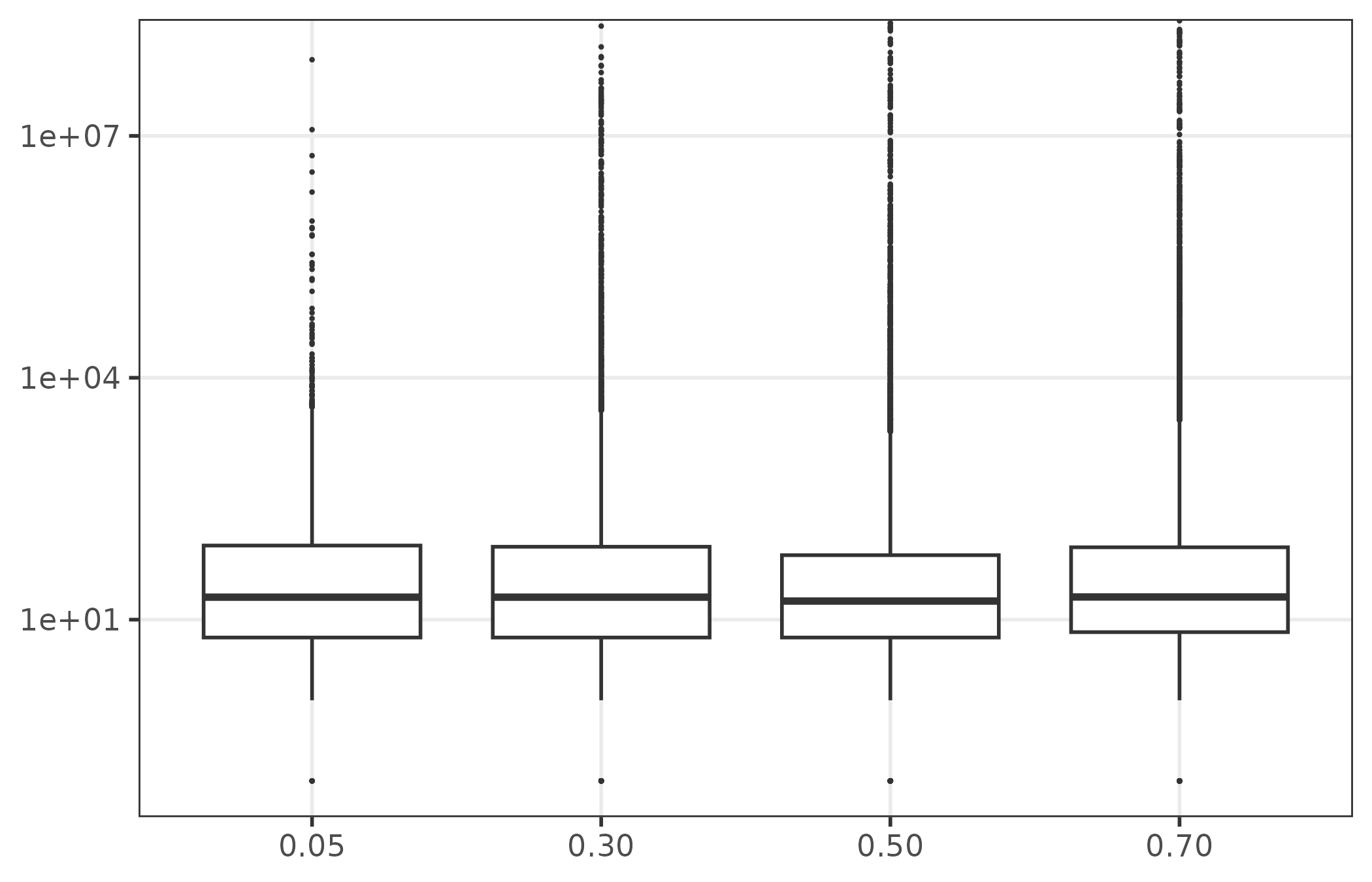} & \includegraphics[width=.45\textwidth, height=.15\textheight]{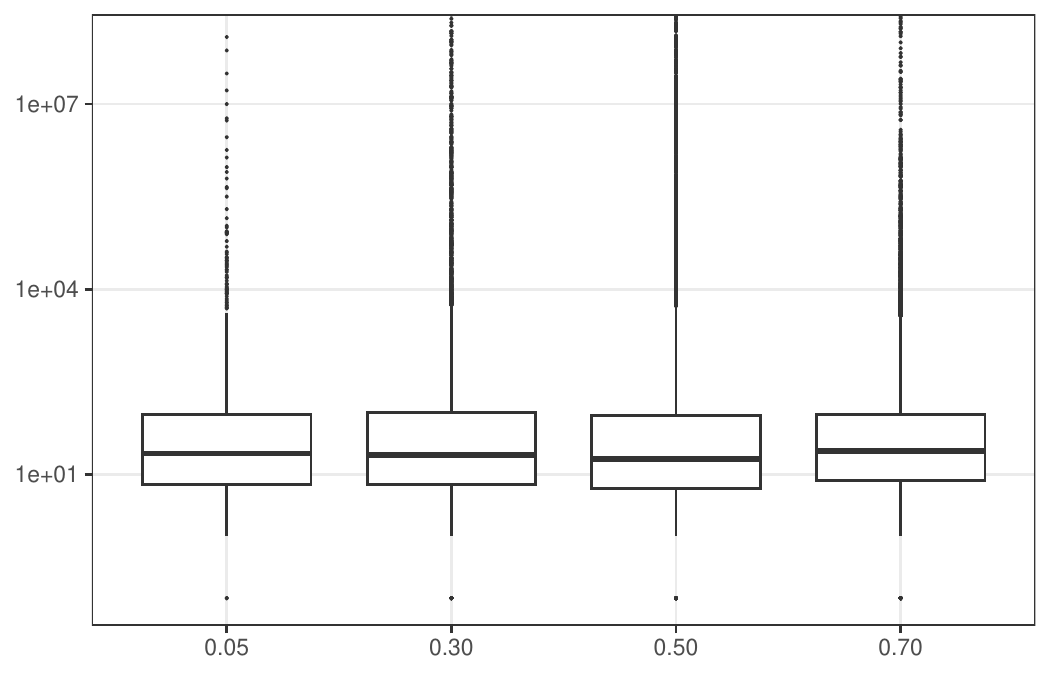}
    \end{tabular}
    \caption{
    Predictions under the MAR time--site missingness scenario for four missing rates (5\%, 30\%, 50\%, 70\%).
    \textbf{Top left:} ratio $|Y_{ij}-\widetilde{Y}_{ij}(X^{\covar}, \ZIPLNPCA)|\,/\,|Y_{ij}-\widetilde{Y}_{ij}((X^{\full}, \ZIPLNPCA)|$.
    \textbf{Top right:} ratio $|Y_{ij}-\widetilde{Y}_{ij}(X^\full, \ZIPLNPCA)|\,/\,|Y_{ij}-\widetilde{Y}_{ij}(X^F, \PLNPCA)|$.
    \textbf{Bottom left:} width of conditional prediction intervals from the \ZIPLNPCA model for the MAR site and years scenario of missing data.
    \textbf{Bottom right:} width of conditional prediction intervals from the \ZIPLNPCA model for the MCAR scenario of missing data.
    \
    }
    \label{fig:compare}
\end{figure}

The top right panel of Figure \ref{fig:compare} illustrates the benefit of modelling zero-inflation explicitly : prediction errors are smaller for the \ZIPLNPCA model than for its non–zero-inflated counterpart (\PLNPCA). By better capturing the large 
proportion of zeros present in the data set,  \ZIPLNPCA provides more accurate imputations across
all levels of missingness.\\


We now focus on prediction intervals, starting by assessing the coverage properties of the proposed prediction intervals. Because coverage is a theoretical probability, empirical coverage estimates vary across simulation replicates. Assuming correct calibration, this variability is governed by a binomial distribution, whose quantiles provide a natural reference for assessing whether observed coverage rates are compatible with the nominal level. The coverage rate falling within the binomial bands per scenario and the rate of missing data are presented in Appendix~\ref{app:cover_rate}. For the 10 replicates, empirical coverage rates of the conditional prediction intervals were generally compatible with the nominal $90\%$ level, although their stability varied across missing-data mechanisms. For all missing data scenarios, for a missing data rate greater than or equal to 0.30, out of 10 replicates, 9 or 10 fall within the binomial bands. In the case of missing data rates of 0.05, coverage rates may be higher than the rate predicted by the upper binomial band. For instance, for the MAR time and MAR time-site scenarios, the average coverage rates are respectively 0.938 and 0.933, compared to upper binomial bound of about 0.93. 


Moving on to prediction interval widths. The bottom left panel of Figure \ref{fig:compare} shows that the distribution of interval widths is broadly similar across missingness levels: the median width remains modest, typically of order 10 to 100, suggesting that most imputations are associated with reasonably tight uncertainty bounds. However, all scenarios exhibit an extremely long tail, with some intervals reaching unrealistic values above $10^{8}$. These large intervals are driven by the substantial variance in the original count data and by the sensitivity of PLN models to high variability, which can lead to rapidly increasing uncertainty. 

We present here conditional prediction intervals rather than marginal ones, as they yielded narrower intervals. However, for high rates of missing data, the widths of conditional and marginal intervals become equivalent, and it may be appropriate to use marginal prediction intervals, which take less time to calculate than conditional ones.


\FloatBarrier

\subsection{North African waterbird}

\subsubsection{Missing data imputation}

We now analyze the monitoring of waterbirds in North Africa, for which much of the data is missing. We focus on rare species~B. Its global conservation status on the  International Union for the Conservation of Nature (IUCN) Red List is \textit{Vulnerable}. 

\begin{figure}[H]
\centering

\subfloat[
Presence and missing data of species~B at 479 wetland sites in North Africa (1990--2023).
Rows correspond to sites and columns to years. MA: Morocco, DZ: Algeria, TN: Tunisia,
LY: Libya, EG: Egypt. 
\label{fig:heatmap}
]{
\includegraphics[width=0.48\textwidth]{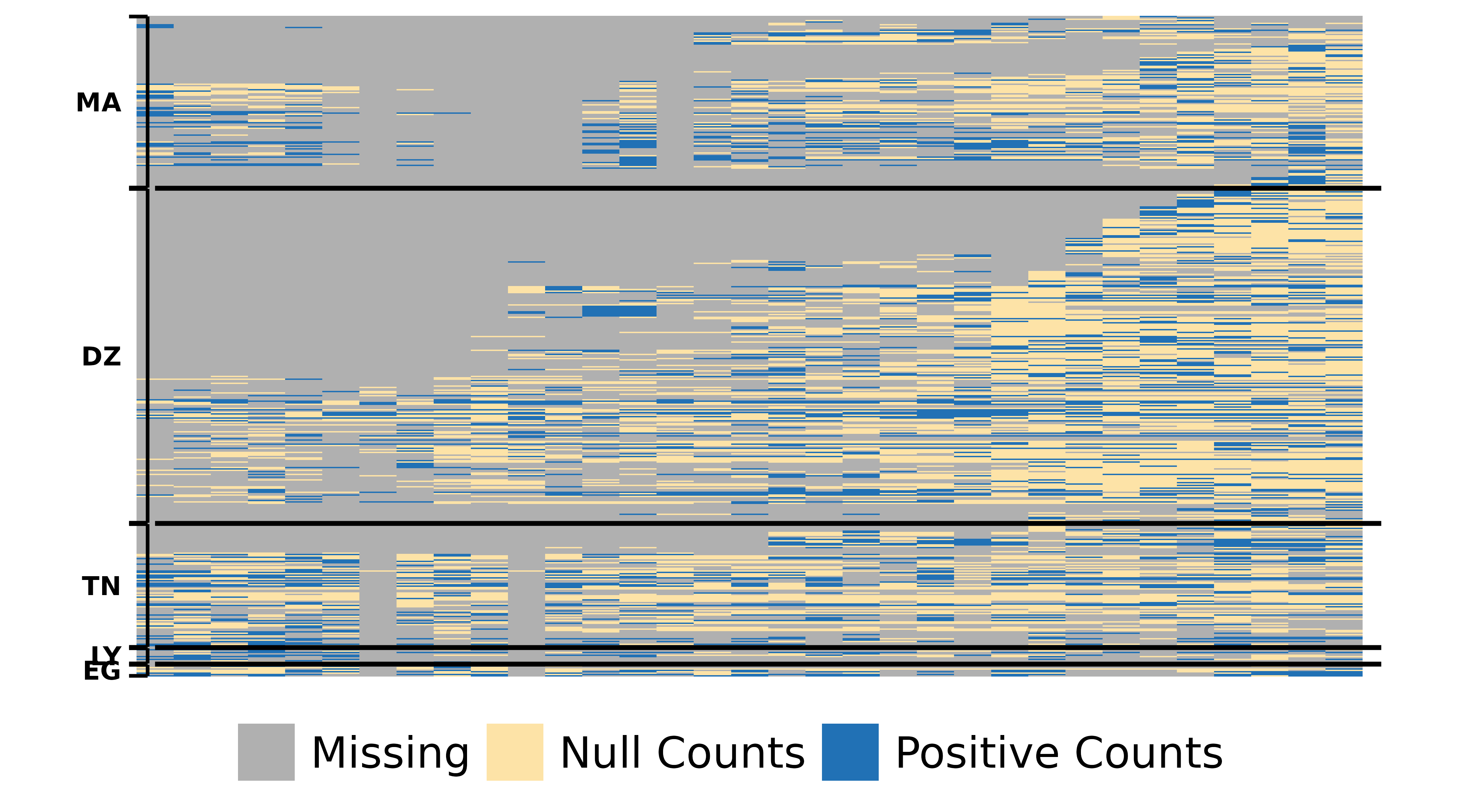}
}
\hfill
\subfloat[
Distribution of the observed abundances of species~B on 479 sites in North Africa
(1990--2023), shown as a histogram of $\log_{10}(1+\text{count})$.
\label{fig:hist}
]{
\includegraphics[width=0.48\textwidth]{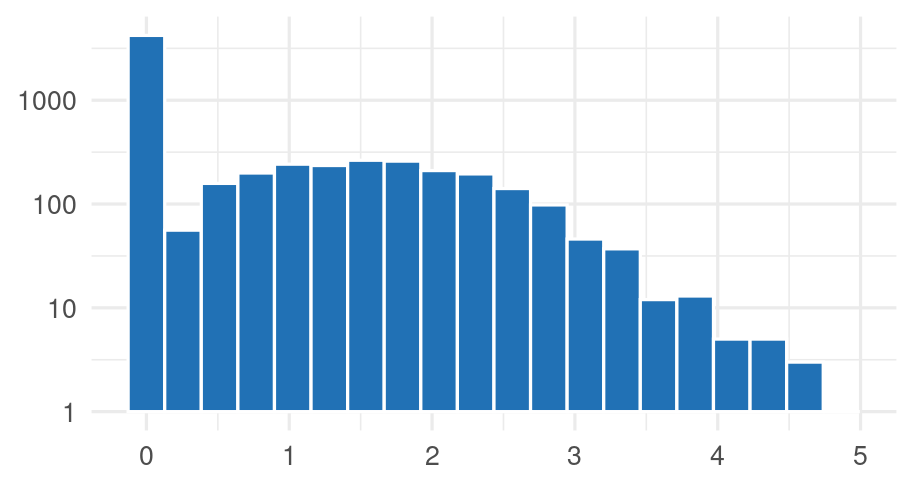}
}

\caption{
Overview of data structure and abundance distribution for species~B across North Africa (1990--2023).
(\textbf{a}) Spatio-temporal distribution of presence records and missing data across wetland sites.
(\textbf{b}) Distribution of observed abundances on a log$_{10}$ scale for both axis.
}
\label{fig:data_overview}
\end{figure}

Our data set comprises abundances collected over a 32-year period (1990–2023) from 479 wetland sites distributed across five North African countries (Morocco, Algeria, Tunisia, Libya, and Egypt), all of which recorded at least one observation of species~B during this period. 
The data set exhibits substantial incompleteness, with 61\% of entries missing, however, it satisfies the identifiability conditions since all pairs of years have been observed at least once at one site and the matrix $X^O$ is full rank. 
As illustrated in Figure~\ref{fig:heatmap}, the missing-data structure most closely corresponds to a  missing-at-random (MAR) mechanism jointly depending on year and site. This pattern reflects heterogeneous monitoring effort across space and year. \\
Among the observed values, 66\% are zeros, indicating pronounced zero inflation, which characterizes the rarity of species~B. 
This leaves us with $(1 - 0.66)(1-0.61) = 13\%$ of  positive observed counts.
The mean count per site and per year is 129.84 individuals, but the distribution is highly skewed, with a standard deviation of 1440.988 (see Figure~\ref{fig:hist}). 

Our first objective was to reconstruct the large proportion of missing values in the data set. 
To this aim we fitted the \ZIPLNPCA model with site and year effects and varying latent dimension. We selected the latent dimension using the BIC criterion defined in Equation \eqref{eq:BIC}, yielding $q=16$. 
The ICL citerion defined in Equation \eqref{eq:ICL} turned out to select the same dimension. 
This provides us with an estimate of the presence probability in each site for each year.
Then, we imputed the missing entries using conditional predictions \eqref{eq:condPred}. 
This imputation step enabled us to reconstruct the complete abundance matrix. 
These reconstructed abundances are plotted in Figure~\ref{fig:map_count} for the most recent year (2023). The corresponding estimated probabilties of presence are also shown in Figure~\ref{fig:map_pres} in Appendix~\ref{app:map_pres}.\\

We also computed the prediction interval associated with each imputation using Algorithm \ref{algo:condPI}.
We present a boxplots of the width on interval widths in Appendix~\ref{app:stat_width}.  
The way we constructed the prediction intervals implies that for all rates of missing data, the 10\%-quantile is equal to 0. 
The median and the 90\%-quantile are of the order of $10^2$. On the other hand, means and standard deviations are very large, which is related to certain very large interval widths and asymmetric intervals. 
For example, in 2023, prediction intervals for almost all imputations range between 0 and 500, with a single notable exception reaching values on the order of $10^{6}$. 
This extreme interval corresponds to the site exhibiting the highest mean and variance in the data set (mean $=17\,149$, standard deviation $=17\,333$), whereas all other sites have means and standard deviations below $5000$. 
In addition, the median of the Monte Carlo predictions, used to calculate the prediction intervals, for this missing value is $19\,188$, and its 80\%-quantile reaches $9\,845$, further illustrating the heavy-tailed behavior of Poisson log-normal models in highly variable settings.


\begin{figure}[H]
    \centering
    \includegraphics[width=0.8\textwidth, trim=30 30 0 30, clip=]{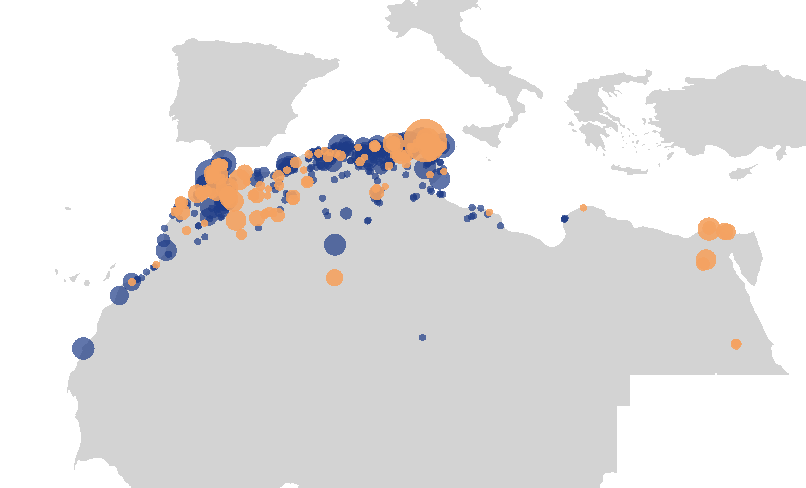}
    \caption{Map of the estimated population of species~B in 2023 across 419 North African sites.\\
    Blue points indicate observed counts, whereas orange points correspond to conditionally imputed values.}
    \label{fig:map_count}
\end{figure}

\subsubsection{Temporal trend estimation \label{sec:trend}} 


\paragraph{Method} Estimating trends is a recurring objective of bird population monitoring. We show now how the \ZIPLNPCA model can be used to estimate a temporal trend and assess its significance. We describe the procedure for the trend of the abundance (i.e. Poisson) part of the \ZIPLNPCA model: a trend for the presence (i.e. Bernoulli) part can be derived in the same way.

Consider the \ZIPLNPCA Model \ref{mod:ZIPLNPCA}, which includes sites and year effect, using $X^{\full}$ as described is section~4.1, that is 
\begin{equation} \label{eq:PLNtimespace}
    Y_{ij} \mid  U_{ij} = 1, Z_i \sim \Pcal(\exp(\beta_0 + \beta^{\text{C}}_j + \beta^{\text{R}}_i + (x_{ij}^{\text{E}})^\top \beta^{\text{E}} + Z_{ij}))  
\end{equation}
where $\beta^{\text{C}}_j$ is the effect of year $j$, $\beta^{\text{R}}_i$ the effect of site $i$ and $(x_{ij}^{\text{E}})^\top \beta^{\text{E}}$ represents the effects of the site-year covariates,
with $\beta^{\text{C}}_1   = \beta^{\text{R}}_1 =0$ to ensure identifiability.  
The variance of the estimated time effect $\widehat{\Var}_n(\widehat{\beta}{\text{C}})$ can be evaluated as described in Section \ref{sec:varTheta}. Note that the first year is included in the upcoming analysis, setting $\widehat{\beta}_1^\text{C}$ and $\widehat{\Var}_n(\widehat{\beta}_1^\text{C})$ to zero, so that $\widehat{\beta}^\text{C}$ has dimension $p$ and $\widehat{\Var}_n(\widehat{\beta}^\text{C})$ has dimension $p \times p$.

To estimate a linear trend, we may now fit a linear regression to the estimated time effects, that is to pose $\widehat{\beta}^{\text{C}}_j = \tau_0 + \tau_1 j + \varepsilon_j$.
Because the least-squares estimates of the intercept $\tau_0$ and slope $\tau_1$ have the close-form expression.
\begin{equation} \label{eq:linTrend}
[\widehat{\tau}_0 \; \widehat{\tau}_1]^\top 
= (X_{\text{time}}^{\!\top} X_{\text{time}})^{-1} X_{\text{time}}^{\!\top} 
\, \widehat{\beta}^{\text{C}},
\end{equation}
where $X_{\text{time}} = [1_p \; T]$, $1_p$ being the vector of dimension $p$ filled with ones, and $T$ having the same dimension and entries $1, 2, \dots p$, we may compute the variance of these estimates as
$$
\Var([\widehat{\tau}_0 \; \widehat{\tau}_1]^\top)
= (X_{\text{time}}^{\!\top} X_{\text{time}})^{-1} 
X_{\text{time}}^{\!\top} \Var(\widehat{\beta}^\text{C}) 
\, (X_{\text{time}} (X_{\text{time}}^{\!\top} X_{\text{time}})^{-1})^\top.
$$

This variance estimate allows us both to test the statistical significance of the slope and to construct a confidence interval for it. A trend for the probability of presence can be estimated (and tested) in the same way, replacing the $\beta$ from Equation \eqref{eq:PLN1} in Model \ref{mod:ZIPLNPCA} with the $\gamma$ of Equation \eqref{eq:PLN2}.

\begin{figure}[H]
\centering

\subfloat[
Temporal trend in abundance with estimated year effects (1990--2023).
\label{fig:trend_imputed}
]{
\includegraphics[width=0.48\textwidth]{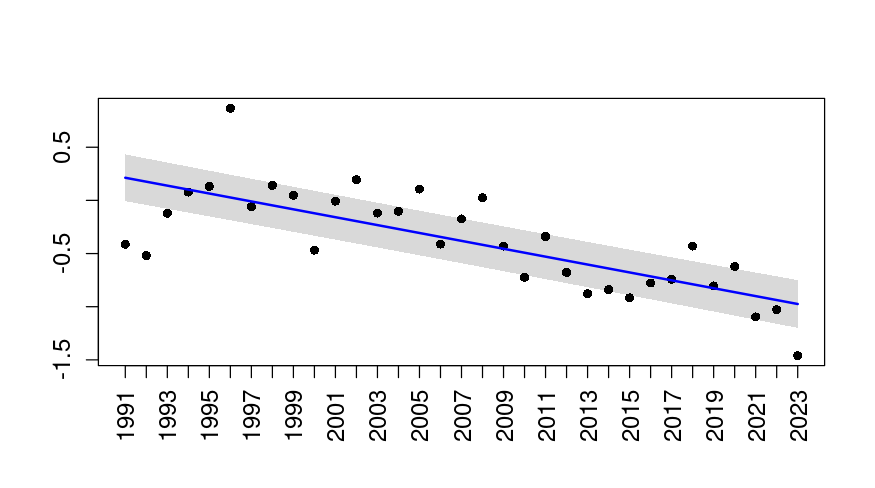}
}
\hfill
\subfloat[
Temporal trend in probability of presence with estimated year effects (1990--2023).
\label{fig:presence_trend}
]{
\includegraphics[width=0.48\textwidth]{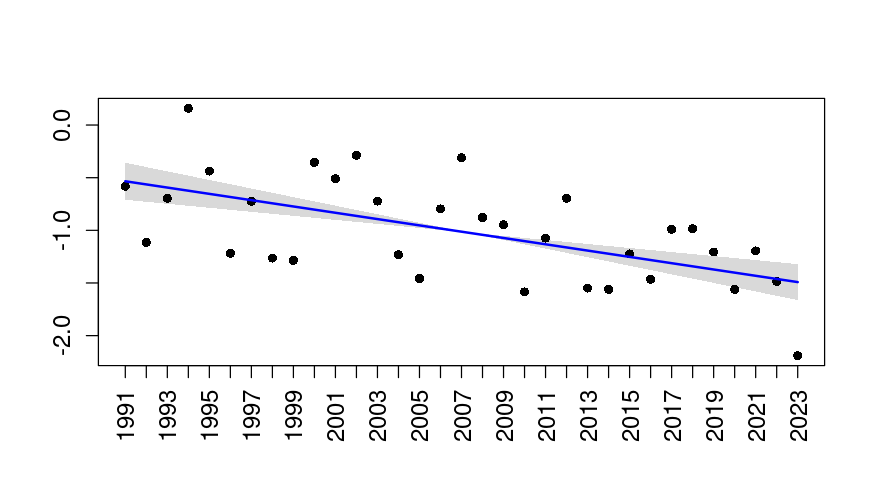}
}

\caption{
Temporal trends in abundance and probability of presence of species~B across North Africa
(1990--2023), as estimated from the fitted model.
}
\label{fig:model_outputs_north_africa}
\end{figure}

\paragraph{Result} In the case of species~B, we identified a significant linear relationship between time and the estimated year effects for both the abundance and the probability of presence components of the model over the period 1990-2023. More precisely, we revealed a significantly negative slope in both cases, with p-values lower than $10^{-7}$. The estimated slope was $-0.04$ for abundance (standard $= 0.023$) and $-0.03$ for the probability of presence (standard error $= 0.006$).\\

\paragraph{Interpretation}
The combination of Models \eqref{eq:PLNtimespace} and \eqref{eq:linTrend} yields that the abundance is expected to be multiplied by a factor $\exp(\tau_1)$ each year.
In the case of species~B, this result indicates a marked decline in abundance and probability of presence. Abundance is multiplied by $\exp(-0.04) \approx 0.96$, which means that abundance decreases by approximately $4\%$ per year. Using a similar reasoning for the probability of presence, we find that it decreases by approximately $3\%$ per year.

\paragraph{Further use.} 
Alternatively, linear temporal trends can be estimated using a single temporal covariate corresponding to years (that is to encode the time effect into a single vector with entries 1990, 1991, \dots, 2023). 
However, the approach we propose offers several advantages. 
We showed that year effects constitute the most informative time covariates for imputation; using them therefore allows a unified framework in which the same temporal structure is exploited both for imputation and for trend estimation. 
Also, although we provide an estimate of the (linear) trend, we do not assume that the effect of time is linear, which could distort the estimation of all other effects. 
An additional interest of the use of a year-specific effect is presented in the next section.

\subsubsection{Change in the trend \label{sec:changepoint}} 


Not all species exhibit strictly linear temporal dynamics. 
Some populations may increase during one period and decrease during another, or vice versa.
By relying on estimated year effects, our approach allows the detection of potential changes in temporal trends, which would not be possible when estimating a single linear trend through a year covariate. 
We illustrate how to detect such a break in the trend using the monitoring of a third species (designated by C) in North Africa.

We use the same notations as in Section \ref{sec:trend} for the estimation of the trend. For a change-point occurring on year $t$, we consider the following segmented regression model:
\begin{equation} \label{eq:changepoint}
\widehat{\beta}_j^{C}
= \tau_0^t + \tau^t_1 j + \tau^t_2 \,   \mathbb{I}\{j > t\}  (j - t) + \varepsilon_j,
\end{equation}
where $\mathbb{I}\{A\}$ is one if $A$ is true, and zero otherwise. 
Let $D_t$ be the vector of dimension $p$, whose first $t$ coordinates are zero and whose subsequent coordinates are equal to $1, 2, \dots, p-t$. We define the design matrix $X^t_{\text{time}} = [1_p \; T \; D_t]$, so that model \eqref{eq:changepoint} can be written as $\widehat{\beta}^{C}=X^t_{\text{time}}\tau+\varepsilon$, with $\tau^t = [\tau^t_0 \; \tau^t_1 \; \tau^t_2]^\top$. 

Once again, the ordinary least squares estimator of $\tau^t$ admits the closed expression $
\widehat{\tau}^{\,t} = ( X_{\text{time},t}^{\top} X_{\text{time},t} )^{-1} X_{\text{time},t}^{\top} \widehat{\beta}^{C}$.
As in Section \ref{sec:trend}, we have access to the asymptotic variance of $\widehat{\tau}^t$ (by replacing $X_{\text{time}}$ with $X^t_{\text{time}}$ in Equation \eqref{eq:linTrend}) and we can test the null hypothesis $\{\tau^t_2 = 0\} $ to assess the significance of a change on year $t$. 
Let $p^t$ be the p-value associated with this test. 
The changepoint $\delta$ can then be estimated by
$$
\widehat{\delta} = \underset{t \in \{2, \dots, p-2\}}{\text{argmin}} p^t.
$$
The significance of the change can be assessed using Bonferroni correction, i.e., by comparing the corrected p-value $(p-2) p^{\widehat{\delta}}$ with the nominal level of each individual test (e.g., $\alpha = 5\%$).


\begin{figure}[H]
    \centering
    \includegraphics[width=0.85\textwidth]{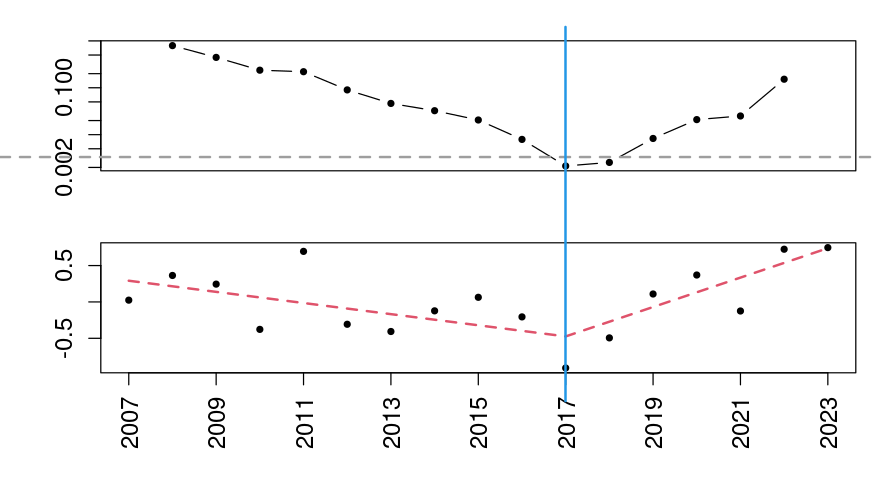}
    \caption{
    Change-point detection for species~C. 
    Top: Wald test p-values (log-scale) for each candidate change-point year; the dashed grey horizontal line represents the Bonferroni-corrected significance threshold.
    Bottom: regression fitted to $\widehat{\beta}^{C}$, with the estimated change-point indicated by the vertical blue line. \\
    \label{fig:changepoint_pvalue_trend}}
\end{figure}

Figure~\ref{fig:changepoint_pvalue_trend} displays the results of this procedure for the monitoring of species~C in North Africa for which a significant changepoint is estimated on year $\widehat{\delta} = 2017$.
We estimate a slope of $\widehat{\tau_1} =-0.08$  between 2006 and 2017 
and a slope of $\widehat{\tau_1} + \widehat{\tau_2}=0.2$  between 2017 and 2023. 
We thus conclude that species~C experienced an annual decline of approximately $8\%$ 
per year between 2006 and 2017, followed by an annual increase of approximately $20\%$ 
from 2017 onwards.


\FloatBarrier

\section{Discussion}



We have presented a method for inferring a log-normal Poisson model with zero inflation, adapted to missing data, for imputing missing observations. 
This architecture, building upon the PLN-PCA model, is adapted to species population monitoring, where both zeros and missing data are abundant, especially when dealing with rare species.
In addition, we provide an estimate of the parameter variance, which enables the construction of prediction intervals. 
These intervals improve uncertainty quantification.
All these tools are compiled in the colvR package available on CRAN.

The framework we propose serves as a valuable tool for estimating population trends accross time, particularly for rare or threatened species for which conservation status is a critical concern. 
We have shown that our model, adjusted for year effects, frees us from the assumption of linearity in the trend and gives us the space to treat the latter as a separate variable. 
We have also illustrated how the information resulting from this modelling allows to address more complex questions, such as the detection of a change in the trend of the population size.
Finally, although adjusting the model with covariates is less effective than adjusting it with site and year effects in terms of imputation accuracy, it still allows us to evaluate the effect of covariates on abundance and to test their significance.

One limit of the model lies in the way the probability of presence is calculated. 
Unlike abundance, probability of presence does not benefit from a latent variable that can account for what is not taken into account in the covariates and to borrow information from one year to another. 
In a preliminary work, we considered a version of the proposed \ZIPLNPCA model including a latent layer in Bernoulli part. We addressed the problem by attempting to model presence alone with a latent Gaussian variable, but we were unable to achieve convergence of a VEM algorithm with this model. This explains why we did not include a latent variable in the presence part of our model. 
We also attempted to use the same latent variables for the presence and the abundance part, but we did not observe any improvement in the prediction of zeros with these variables. 

More importantly, the quality and explanatory power of the covariates play a critical role. The more variance the covariates can explain, the more reliable and accurate the predictions will be.
This underlines the importance of selecting relevant covariates in ecological modeling, especially when imputation is at stake. Although site and year effects turn out to be more relevant, when imputation is at stake, including relevant and interpretable covariates is critical to better understand ecological mechanisms, to distinguish the main environmental drivers for the species conservation, or to assess the efficiency of a given conservation policy.

\paragraph{Acknowledgements.} 
The authors thank the national coordinators of the International Waterbird Census for kindly making their winter waterbird count data available, allowing us to confirm the practical application of the method using their datasets. Many thanks to Mr. Marco Zenatello and Dr. Nicola Baccetti, Italian National Coordinators for ISPRA (Istituto Superiore per la Protezione e la Ricerca Ambientale), the ONCFS/FNC/FDC Wetlands Waterbird Network as well as the Mediterranean Waterbirds Network and its coordinators: Dr. Mohamed Dakki, Moroccan National Coordinator for GREPOM (Groupement pour la Protection des Oiseaux au Maroc); Dr. Khaled Etayeb, Libyan National Coordinator for LSB (Libyan Society for Birds); Ms. Nadjiba Bandjedda and Mr. Mohamed Samir Sayoud, Algerian National Coordinators for DGF (Direction Générale des Forêts, Algeria); and Mr. Hichem Azafzaf, Tunisian National Coordinator for AAO/BirdLife Tunisia (Association des Amis des Oiseaux/BirdLife Tunisia). We are grateful to the INRAE MIGALE bioinformatics facility (MIGALE, INRAE, 2020. Migale Bioinformatics Facility, doi: 10.15454/1.5572390655343293E12) for providing help and/or computing and/or storage resources. This work was founded by the Tour du Valat (Research Institute for the Conservation of Mediterranean Wetlands), the Agence Française de Développement and the institut des Mathématiques pour la Planète Terre.

\bibliographystyle{plainnat}
\bibliography{bibliography}

\appendix
\clearpage

\onecolumn

\renewcommand{\thefigure}{A.\arabic{figure}}
\renewcommand{\thetable}{A.\arabic{table}}
\renewcommand{\thealgorithm}{\arabic{algorithm}}
\setcounter{figure}{0}
\setcounter{table}{0}


\section{Moments of the PLN distribution}\label{proof:propmoments}
 
\begin{proof}
Let us define $$\mu_{ij} = x_{ij}^T \beta$$ and $\lambda_{ij} = \mu_{ij} + Z_{ij}$.
When the species is present $(U_{ij}=1)$, Model \ref{mod:ZIPLNPCA} inherits the general properties of the Poisson log-normal model given by \citet{AiH89}, namely
\begin{align*} 
    \text{First moment:} & & E^+_{ij} := \mathbb{E}[Y_{ij} \mid U_{ij} = 1] 
    & = \mathbb{E}[\exp(\lambda_{ij})] = \exp(\mu_{ij} + \tfrac{1}{2} \sigma_j^2), \\
    \text{Second moment:} & &
    \mathbb{E}[Y_{ij}^2 \mid U_{ij} = 1] 
    & = \mathbb{E}[\exp(\lambda_{ij}) + \exp(2\lambda_{ij})] 
    = \exp(\mu_{ij} + \tfrac{1}{2} \sigma_j^2) + \exp(2\mu_{ij} + 2\sigma_j^2),\\
    \text{Variance:} & &  V^+_{ij}   := \Var(Y_{ij} \mid U_{ij} = 1) &= E^+_{ij} + (E^+_{ij})^2 (e^{\sigma_j^2} - 1 ), \\
  \text{Third moment:} & &
   \mathbb{E}[Y_{ij}^3 \mid U_{ij} = 1] 
 & = \mathbb{E}[\exp(\lambda_{ij}) + 3\exp(2\lambda_{ij}) + \exp(3\lambda_{ij})] \\
   & & & = \exp(\mu_{ij} + \tfrac{1}{2} \sigma_j^2) + 3\exp(2\mu_{ij} + 2\sigma_j^2) + \exp(3\mu_{ij} + \tfrac{9}{2}\sigma_j^2)\\
   \text{Covariance:} & &  C^+_{ijk} &:= \Cov(Y_{ij}, Y_{ik} \mid U_{ij} U_{ik} = 1)  = E^+_{ij} E^+_{jk} (e^{\sigma_{jk}} - 1) & 
\end{align*}
Accounting for the possible absence of the species $\mathbb{P}(U_{ij}=1)= \pi_{ij}$ yields the forms for the unconditional moments of Proposition \ref{prop:moments} and of the demonstration of Proposition \ref{prop:identif1}.





\end{proof}

\section{Additional formulas and algorithms}

\subsection{Second order derivatives of the ELBO} \label{app:secondDeriv}

We give here the second order derivatives of the contribution $J_i(\theta, \pt)$ of the $i$-th to the ELBO, defined in Equation \eqref{eq:elboi}. The second derivatives are for the model parameter $\theta \gamma, \beta, C)$ are
\begin{align*}
    \partial^2_{\gamma, \gamma^\top} J_i 
    & =\sum_j \Omega_{ij} \probU_{ij} (\probU_{ij} -1) x_{ij} x_{ij}^\top, \\
    \partial^2_{\beta, \beta^\top} J_i 
    & =- \sum_j \Omega_{ij} \xi_{ij} A_{ij} x_{ij} x_{ij}^\top, \\
    \partial^2_{C_j, C_j^\top} J_i 
    & = - \Omega_{ij} \xi_{ij} A_{ij} \left((m_i + (C_j \odot s_i))(m_i + (C_j \odot s_i))^\top +
    S_i\right), 
\end{align*}  
the cross-second derivatives $\partial^2_{\beta \gamma} J_i$ and $\partial^2_{\gamma C_j} J_i$ are zero and
\begin{align*}
    \partial^2_{\beta C_j^\top} J_i & = - \Omega_{ij} \xi_{ij} A_{ij} x_{ij} (m_i + C_j \odot s_i)^\top .
\end{align*}  
Regading the variational parameters $\psi = (M, S)$, the second derivatives are
\begin{align*}
    \partial^2_{m_i, m_i^\top} J_i & = - \sum_j \Omega_{ij} \xi_{ij} A_{ij} C_j C_j^\top - I_q, \\
    \partial^2_{s_i, s_i^\top} J_i & = - \frac12 \left(S_i^{-2} + \sum_j \Omega_{ij} \xi_{ij} A_{ij} (C_j \odot C_j)(C_j \odot C_j)^\top\right), \\
    \partial^2_{\xi_{ij}, \xi_{ij}} J_i & = - \frac{1}{\xi_{ij}(1 - \xi_{ij})},
\end{align*}  
and the second cross-derivatives are 
\begin{align*}
    \partial^2_{m_i, s_i^\top} J_i & =- \frac12 \sum_j \Omega_{ij} \xi_{ij} A_{ij} C_j (C_j \odot C_j)^\top, &
    \partial^2_{m_i, \xi_{ij} } J_i & =\Omega_{ij} (Y_{ij} - A_{ij}) C_j, \\
    \partial^2_{s_i, \xi_{ij}} J_i & =- \frac{1}{2} \Omega_{ij} A_{ij} (C_j \odot C_j).
\end{align*}  
Finally, the second cross-derivatives between model and variational parameters are
\begin{align*}
    \partial^2_{\beta, m_i^\top} J_i & =- \sum_j \Omega_{ij} \xi_{ij} A_{ij} x_{ij} C_j^\top, &
    \partial^2_{\beta, s_i^\top} J_i & =- \frac12 \sum_j \Omega_{ij} \xi_{ij} A_{ij} x_{ij} (C_j \odot C_j)^\top, \\
    \partial^2_{\beta \xi_{ij} } J_i & = \Omega_{ij}(Y_{ij} - A_{ij}) x_{ij}, &
    \partial^2_{\gamma \xi_{ij} } J_i & = \Omega_{ij} x_{ij}, \\
    \partial^2_{C_j \xi_{ij} } J_i & = \Omega_{ij} \left((Y_{ij} - A_{ij}) m_i - A_{ij} (C_j \odot s_i)\right),
  \end{align*}  
all others second cross-derivatives being zero.



\subsection{Marginal prediction}  \label{app:margPred}

\begin{algorithm}[H]
\caption{Marginal predictions and prediction intervals}
\label{algo:margPI} ~
    \begin{description}
        \item[{\sl Input:}] estimates $\widehat{\theta}$ and $\widehat{\Var}(\widehat{\theta})$, number of particles $B$. 
        \item[{\sl Output:}] conditional intervals for the $\Esp(Y_{ij} \mid Y^O)$ and conditional prediction interval for the $Y_{ij}$. 
        \item[{\sl Algorithm:}] ~
        \begin{enumerate}
            \item For $b = 1 \dots B$, 
            \begin{enumerate}
                \item Sample $\theta^b \sim \Ncal\left(\widehat{\theta}, \widehat{\Var}(\widehat{\theta})\right)$, denote $\gamma^b$, $\beta^b$ and $\Sigma^b$ the resulting model parameters and compute
                $$
                \widehat{Y}_{ij}
                = \widehat{\pi}_{ij} \exp\left(x_{ij}^\top \widehat{\beta} + \widehat{\sigma}_{jj} / 2\right).
                $$
                \item For each site $i$, such that $\Mcal_i \neq \varnothing$, sample independently $W^b_i \sim \Ncal(0, \Sigma^b_i)$, compute $Z_i^b = C^b W_i^b$, and for each year $j \in \Mcal_i$, sample $U^b_{ij} \sim \Bcal(\pi_{ij}^b)$, where $\logit(\pi_{ij}^b) = x_{ij}^\top \gamma^b$, set $Y^b_{ij} = 0$ if $U^b_{ij} = 0$ and, otherwise, sample independently 
                $$
                Y^b_{ij} \sim \Pcal\left(\exp\left(x_{ij}^\top \beta^b\right)\right);
                $$ 
            \end{enumerate}
            \item Determine the intervals for each missing  observation $(i, j) \in \Mcal$:
            \begin{itemize}
                \item Use the empirical quantiles of the $(\widehat{Y}^b_{ij})_{1 \leq b \leq B}$ to determine a marginal confidence interval for each conditional expected count $\Esp(Y_{ij})$;
                \item Use the empirical quantiles of the $(Y^b_{ij})_{1 \leq b \leq B}$ to determine a marginal prediction interval for $Y_{ij}$.
            \end{itemize}
        \end{enumerate}
    \end{description}
\end{algorithm}

\section{Additional information and simulations}

\subsection{Covariates} \label{app:covariates}
\paragraph{Covariates associated with the Bird A}
\begin{description}
    \item[Sites covariates:]~Longitude, latitude, humid zone's superficie in squared kilometers, distance to the nearest town in kilometers, distance to the nearest coast in kilometers, average elevation.
    \item[Years covariates:]~Year, winter northwestern european temperature anomaly, spring northwestern european temperature anomaly, winter northeastern european temperature anomaly, spring northeastern european temperature anomaly, mean winter north atlantic oscillation, average spring precipitation in northeastern Europe for year n-1, average spring precipitation in northwestern Europe for year $n-1$.
    \item[Site-year covariates:] Average winter precipitation. 
\end{description}

\subsection{Simulation Study: Results on the species~A dataset} 
\label{app:sim2}

Here we present the results of comparisons between, on the one hand, the model with covariates and the model with site effects and year effects and, on the other hand, the inflated zero model and the non-inflated zero model for all scenarios of artificially added missing data in the species~A abundance dataset.

\begin{figure}[H]
\centering

\subfloat[
\label{fig:ratio_zip_pln}
]{
\includegraphics[
  width=\textwidth,
  height=0.40\textheight,
  keepaspectratio
]{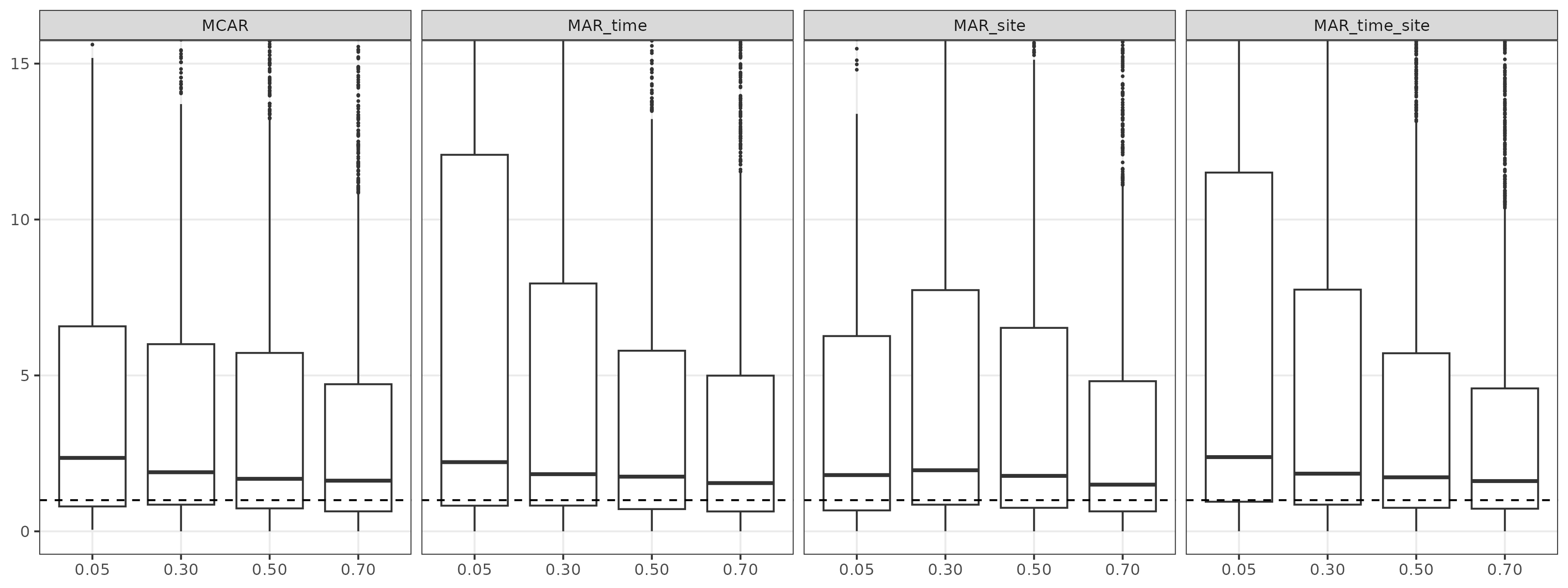}
}

\vspace{0.2cm}

\subfloat[
\label{fig:ratio_cov_x}
]{
\includegraphics[
  width=\textwidth,
  height=0.40\textheight,
  keepaspectratio
]{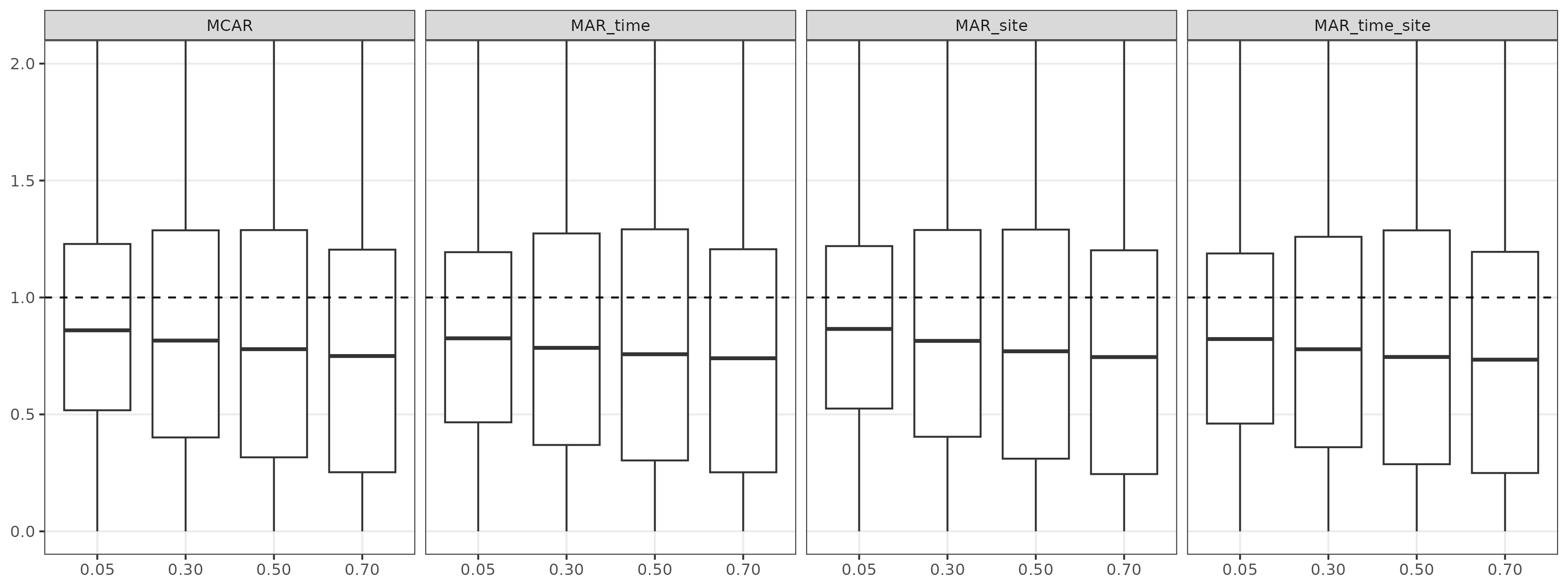}
}

\caption{
Comparison of relative imputation errors across missing-data scenarios and missing rates.
(\textbf{a}) Boxplots of the ratio between imputation errors obtained with the covariate-based model and the site and year effects model ($X^{G} / X^{F}$), computed on missing entries only.
(\textbf{b}) Boxplots of the ratio between imputation errors obtained with \textsc{ZIPLNPCA} and \textsc{PLNPCA}.
In both panels, results are shown for the four missing-data mechanisms and the four missing-data rates, and ratios are computed across the missing entries of the 100 simulation replicates for each configuration.
}
\label{fig:ratio_comparison}
\end{figure}

\subsection{Influence of the proportion of variance explained by covariates on prediction accuracy} \label{app:sim1}

We present here a small simulation study to illustrate the importance of the quality of the covariates, when aiming at accurate imputations.

\paragraph{Simulation design} 

Our goal is to vary, for a fixed total variance, the proportion explained by the covariates versus that explained by the latent structure. 
Specifically, we distributed the total variance between the linear predictor $X\beta$, which corresponds to the variance explained by the covariates, and the latent component governed by $\Sigma$. We made the fraction of variance explained by $\beta$ varied from 0\% to 100\%, by steps of 20\%,which yields six different scenarios.
For each scenario, the marginal moments of the simulated counts were calibrated to match a mean of $20$ and a variance of $5000$. We simulated count matrices with $n = 500$ rows (sites) by $p = 20$ columns (years). 
In addition, for each of these six configurations, we introduced three levels of missing data ($5\%$, $30\%$, and $60\%$) with an MCAR pattern.

\paragraph{Results}
Figures~\ref{fig:grand_panorama} display the true versus predicted values across the six variance allocation scenarios for missing data rates of 5\%, 30\%, and 60\%. Two main results emerge. 

\begin{enumerate}
    \item When most of the variance is explained by the covariates, predictions closely align with the diagonal even at higher missing rates, as seen in the last two lines of the figure, highlighting the stabilizing role of covariates that capture the main signal. 
    \item When covariates explain little: as the latent share increases, prediction scatter increases, as seen in the 2 first lines of the figure, and this effect amplifies as the missing rate rises from 5\% to 60\%.
\end{enumerate}

This simulation study highlights the importance of selecting covariates that best describe the sites and years so that they can account for the largest share of the variance.


\begin{figure}
    \centering
    \includegraphics[
        height=0.90\textheight,
        keepaspectratio,
        trim=1.5cm 1.5cm 1.5cm 1.5cm,
        clip
    ]{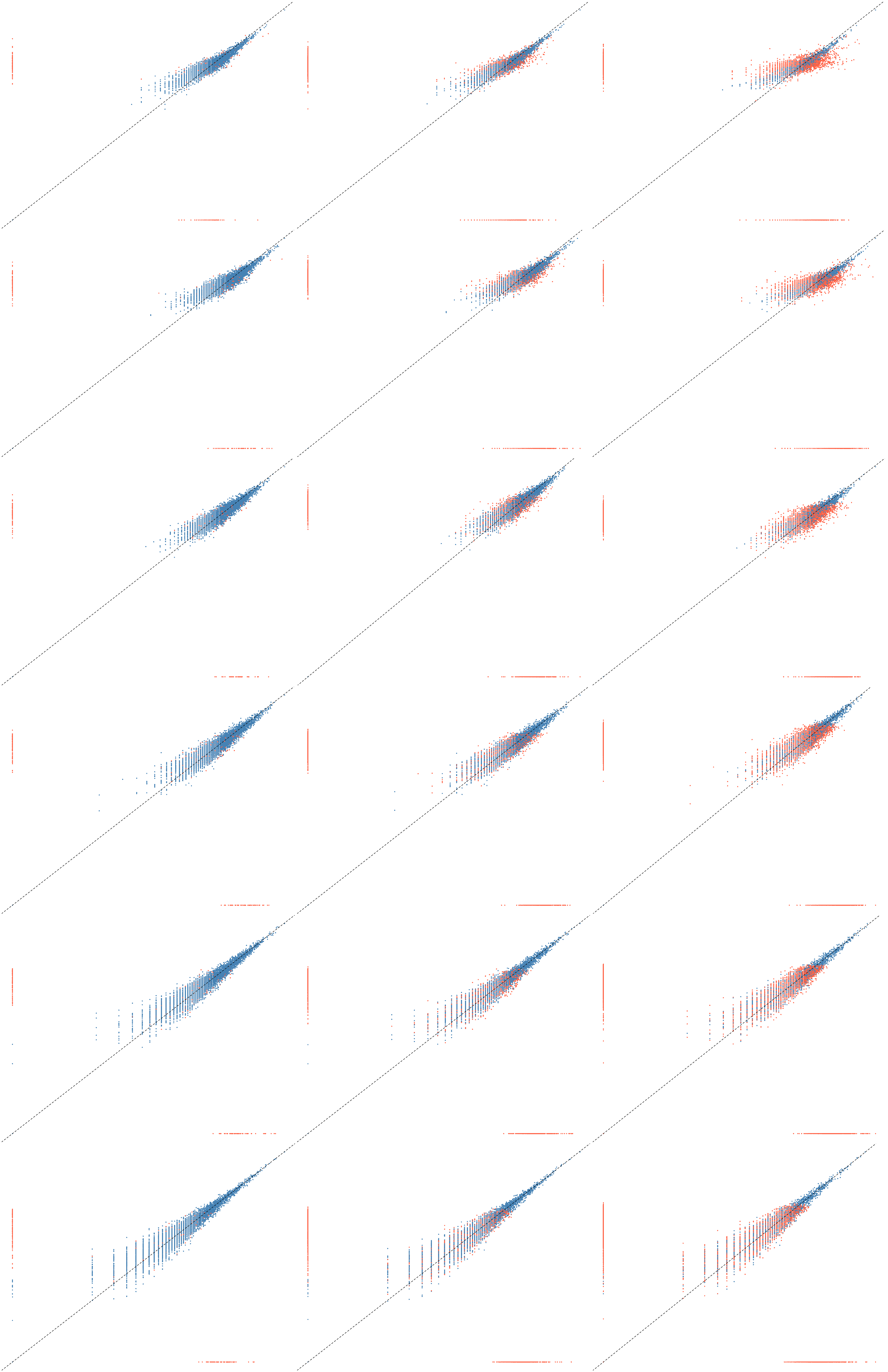}
    \caption{Comparison between true and predicted values across simulation scenarios.
    Columns correspond to increasing proportions of missing data (from left to right: 5\%, 30\%, and 60\%).
    Rows correspond to simulation scenarios with an increasing proportion of variance explained by covariates and a decreasing contribution of latent variables from top to bottom.
    Blue points represent observed data, whereas red points represent missing data.}
    \label{fig:grand_panorama}
\end{figure}


\subsection{Prediction interval coverage rate}
\label{app:cover_rate}

The table reports, for each missing-data mechanism and missing-data rate, the proportion among the 10 simulation replicates whose empirical coverage of the prediction intervals falls within the binomial reference bands associated with the nominal coverage level. More specifically, for a simulation replicate with $M$ missing observations, the number of prediction intervals that actually contain the observed value follows a binomial distribution $\Bcal(M, 0.9)$. We calculated the 2.5\%- and 97.5\%-quantiles of this binomial distribution and show in Table \ref{tab:inside_neff_rate}, for each simulation configuration, the number of replicates for which the proportion of covering intervals lies between these quantiles.

\begin{table}[H]
\centering
\caption{Number of simulation replicates (among 10) whose empirical coverage falls within the binomial bands, for each missing-data mechanism (pattern) and missing-data rate.}
\label{tab:inside_neff_rate}
\begin{tabular}{lcccc}
\hline
\textbf{Pattern} & \textbf{0.05} & \textbf{0.30} & \textbf{0.50} & \textbf{0.70} \\
\hline
MAR site       & 9 & 10 &  9 &  9 \\
MAR time       & 3 &  9 & 10 &  9 \\
MAR time-site & 5 &  9 &  9 &  9 \\
MCAR            & 7 & 10 & 10 & 10 \\
\hline
\end{tabular}
\end{table}

\subsection{Width of the prediction intervals}
\label{app:stat_width}

Here we present descriptive statistics on the widths of prediction intervals for imputations in counts of species 

\begin{table}[H]
\centering
\caption{Summary statistics of prediction interval widths for missing values of $Y$.
Quantiles correspond to those displayed in the boxplot (Q1--median--Q3), with additional
deciles (10\% and 90\%).}
\label{tab:ip_width_summary}
\begin{tabular}{lrrrrrrr}
\toprule
 & q10 & q25 & Median & q75 & q90 & Mean & SD \\
\midrule
Interval width &
0 &
6 &
27 &
119.4 &
390.2 &
494.3 &
11469.3 \\
\bottomrule
\end{tabular}
\end{table}

\subsection{Map of presence probability for Bird B}
\label{app:map_pres}

\begin{figure}[H]
    \centering
    \includegraphics[width=0.8\textwidth, trim=30 30 0 30, clip=]{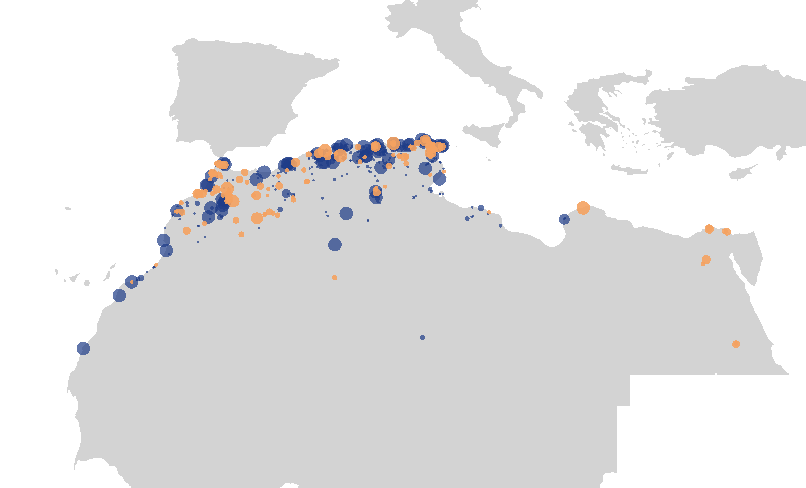}
    \caption{Map of the estimated presence probability of Bird B in 2023 across 419 North African sites.\\
    Blue points indicate observed counts, whereas orange points correspond to conditionally imputed values.}
    \label{fig:map_pres}
\end{figure}

\end{document}